\documentclass[a4paper,10pt]{article}
\usepackage[dvipdfmx]{hyperref}
\usepackage{url,graphicx}
\usepackage{amsthm,amsmath,amssymb,amsfonts}
\usepackage{showexpl}
\usepackage{mathrsfs,xr}
\usepackage[all,arc,curve,color,frame]{xy}
\xyoption{all}

\title{On Type II noncommutative geometry and the JLO character}

\author{Alan Lai
\thanks{Email: alan@math.toronto.edu}
\\University of Toronto
}
\begin{document}

\maketitle

\newcommand{\cch}
{\widetilde{\mathtt{Ch}}} 
\newcommand{\Dd}{\mathcal{D}}
\newcommand{\sDd}{\slash\hspace{-0.26cm}\mathcal{D}}
\newcommand{\dDd}{\Dd_I}
\newcommand{\CCc}{\overline{\Cc}^{\Ggamma}}
\newcommand{\TR}{\operatorname{Tr}}
\newcommand{\CCHERN}{\operatorname{ch}}
\newcommand{\KKEI}{\operatorname{K}}
\newcommand{\g}{\mathfrak{g}} 
\newcommand{\naturalnumber}{\mathbb{N}}
\newcommand{\CAS}{\operatorname{Cas}}
\newcommand{\MAP}{\operatorname{Map}}
\newcommand{\GAU}{\operatorname{Gau}}
\newcommand{\MATRIX}{\operatorname{Mat_\mathbb{C}(N)}}
\newcommand{\END}{\operatorname{End}}
\newcommand{\Uu}{\mathcal{U}}
\newcommand{\Ww}{\mathcal{W}}
\newcommand{\Aa}{\mathcal{A}}
\newcommand{\Bb}{\mathcal{B}}
\newcommand{\Hh}{\mathcal{H}}
\newcommand{\Ss}{\mathcal{S}}
\newcommand{\Ff}{\mathfrak{F}}
\newcommand{\Nn}{\mathcal{N}}
\newcommand{\Pp}{\mathcal{P}}
\newcommand{\Tt}{\mathcal{T}}
\newcommand{\Ll}{\mathcal{L}_{\Nn}}
\newcommand{\Kk}{\mathcal{K}_{\Nn}}
\newcommand{\Mm}{\mathfrak{m}}
\newcommand{\SUPP}{\operatorname{supp}}
\newcommand{\HE}{\operatorname{HE}}
\newcommand{\KY}{\operatorname{K}}
\newcommand{\ED}{\mathfrak{e}}
\newcommand{\VE}{\mathfrak{v}}
\newcommand{\SO}{\operatorname{s}}
\newcommand{\RA}{\operatorname{r}}
\newcommand{\ID}{\operatorname{id}}
\newcommand{\HOM}{\operatorname{Hom}}
\newcommand{\HOL}{\operatorname{Hol}}
\newcommand{\Zz}{\mathbb{Z}}
\newcommand{\Cc}{\mathfrak{C}}
\newcommand{\CL}{\mathbb{C} \operatorname{ l}}
\newcommand{\Hhg}{\mathcal{G}}
\newcommand{\Gg}{\mathcal{G}}
\newcommand{\Ggamma}{\mathbf{\Gamma}}
\newcommand{\KKi}{\mathbf{\chi}}
\newcommand{\Ch}{\mathtt{Ch}_{^{_{\operatorname{JLO}}}}}
\newcommand{\ch}{\mathtt{ch}}
\newcommand{\hCh}{\mathtt{Ch}_{^{_{\operatorname{JLO}}}}}
\newcommand{\Ilim}{\varinjlim}
\newcommand{\Plim}{\varprojlim}
\newcommand{\LG}{\mathcal{LG}}
\newcommand{\ZO}{\diamond}
\newcommand{\KER}{\operatorname{ker}}
\newcommand{\DIFF}{\operatorname{diff}}
\newcommand{\ENDO}{\operatorname{End}}
\newcommand{\DOM}{\operatorname{Dom}}
\newcommand{\KCYCLE}{(\rho,\Nn,\Dd)}
\newcommand{\FREDMOD}{(\rho,\Nn,F)}
\newcommand{\IND}{\mathtt{Ind}_\tau}
\newtheorem{definition}{Definition}[section]
\newtheorem{example}{Example}[section]
\newtheorem{theorem}{Theorem}[section]
\newtheorem{lemma}[theorem]{Lemma}
\newtheorem{corollary}[theorem]{Corollary}
\newtheorem{proposition}[theorem]{Proposition}
\newtheorem{remark}[definition]{Remark}

\setcounter{section}{-1}

\begin{abstract}
The Jaffe-Lesniewski-Osterwalder (JLO) character \cite{jlo} 
is a homomorphism from $\KKEI$-homology to entire cyclic cohomology.
This paper extends 
the domain of the JLO character 
to include Type II noncommutative geometry,
the geometry represented by
  Breuer-Fredholm modules; 
and shows that  the JLO character
defines the same 
cohomology class as the Connes-Chern character
 \cite{type2index}
 in entire cyclic cohomolgoy.
\end{abstract}

\tableofcontents

\section{Introduction}
In non-commutative geometry,
the guiding principle 
is that the topology of spaces 
is encoded in properties of their algebras
of continuous functions.
A theorem of Gelfand-Naimark \cite{invitetoncg} states that 
any commutative unital $C^*$-algebra  is of the form $C(X)$
for some compact Hausdorff space $X$.
 Therefore, the category of $C^*$-algebras (or even more generally
Banach $*$-algebras) is seen as an extension of
the category of compact Hausdorff topological spaces,
and a general $C^*$-algebra 
is sometimes referred to as a non-commutative topological space.
The geometric features on a $C^*$-algebra $A$ are incorporated 
by the concept of an
an unbounded Fredholm module $(\rho,B(\Hh),\Dd)$ over $A$,
where $\rho$ is a continuous representation of $A$
onto the Hilbert space $\Hh$, and $\Dd$ is an unbounded
Fredholm operator on $\Hh$ that satisfies certain axioms.
As the prototypical example, 
let $A$ be the algebra of continuous functions 
on a closed Riemannian
manifold, $\Hh$ the square integrable sections
 of a spinor bundle with its
natural action of $A$, and $\Dd$ the associated Dirac operator.
Geometric features on the   manifold
such as geodesics, dimension, integrations,
and  differential forms etc can be 
retrieved algebraically in terms of $A$, $B(\Hh)$, and $\Dd$ \cite{greenbook}.
Connes gives a set of five axioms
characterizing
 the Fredholm modules arising in this way \cite{reconstruction}.
 Taking $A$ to be
non-commutative thus leads to a notion of a non-commutative manifold.
This theory is summarized in Connes' famous book \cite{redbook},
further details and newer developments are described in 
\cite{greenbook} and \cite{invitetoncg}.

Each Fredholm module assigns
an integer, the Fredholm index, to
a given element in the $\KKEI$-theory of $A$.
The Fredholm index provides a $\mathbb{Z}$-valued pairing
between the $\KKEI$-homology of $A$ and the $\KKEI$-theory of $A$.
 In the commutative setting, the index can be viewed as the
index of the Dirac operator $\Dd$, twisted by a vector bundle.

Suppose the unbounded
Fredholm module $(\rho,B(\Hh),\Dd)$ is finitely summable,
 a condition that models finite dimensionality
according to Connes' axioms,
Jaffe-Lesniewski-Osterwalder \cite{jlo} 
defined a cocycle $\Ch^\bullet$ in 
the entire cyclic cohomology
 $\HE^\bullet(A)$, now known as the JLO character.
Together with 
the $\KKEI$-theory character $\CCHERN_\bullet: 
\KKEI_\bullet(A) \to \HE_\bullet(A)$,
they
 intertwine the 
$\KKEI$-theoretical 
pairing given by 
the Fredholm index with
the cohomological pairing between $\HE^\bullet(A)$
and $\HE_\bullet(A)$
\cite{getzlerodd,getzlereven}.
Such a result was originally established in a more general setting
for
weakly $\theta$-summable
 Fredholm modules, where
weak $\theta$-summability can be thought of a suitable notion of
infinite-dimensionality.
 Consequently, the JLO character provides a formula for the
Fredholm index in terms of entire cyclic (co)homology
for infinite dimensional non-commutative manifolds, which
was the original motivation of JLO's work \cite{jlo}.
 Furthermore,
the formula reduces to the index formula of
Atiyah-Singer in the commutative setting \cite{heatkernel,boundary}.

%

The operator $\Dd$ of a Fredholm module
 plays the role of a Dirac operator,
and is typically unbounded.
 However, there is a
canonical way of passing from an unbounded Fredholm module
to a bounded one $(\rho,B(\Hh),F)$,
 essentially by taking bounded functions of $\Dd$,
 and the latter are often easier to work with in practice.

When the
 bounded Fredholm module
$(\rho,B(\Hh),F)$ is finitely summable,
there is  a character $\mathtt{ch}^n$  due to
Connes \cite{ncdg}, which again is a cocycle in $\HE^\bullet(A)$,
and 
 $\mathtt{ch}^n$
 intertwines the 
 $\KKEI$-theoretical pairing the same way as
the JLO character \cite{ncdg,chernreduct}.
When $(\rho,B(\Hh),F)$ is the associated bounded module of 
a finitely summable unbounded Fredholm module
$(\rho,B(\Hh),\Dd)$,
Connes-Moscovici proved that  in fact
the cocycle $\mathtt{ch}^n(F)$ of 
$(\rho,B(\Hh),F)$
defines the same cohomology class as
 the cocycle $\Ch^\bullet(\Dd)$ of 
$(\rho,B(\Hh),\Dd)$ in $\HE^\bullet(A)$ \cite{chernreduct}.
 
Type II non-commutative geometry 
replaces the algebra $B(\Hh)$ with a (possibly) Type II von Neumann algebra
$\Nn \subset B(\Hh)$, using Breuer's Fredholm theory relative to
the von Neumann algebra $\Nn$ \cite{breuer1,breuer2}.
A Type II non-commutative geometry 
on the algebra $A$ 
is  given by an unbounded Breuer-Fredholm module 
$(\rho,\Nn,\Dd)$ over $A$,
where
 $\Dd$ is a Breuer-Fredholm operator affiliated with $\Nn$
that satisfies certain axioms.
Examples of unbounded Breuer-Fredholm modules 
arise from foliations or geometry with degeneracies.
A number of examples can be found in \cite{type2index}.

Parallel to the Type I setting,
the Breuer-Fredholm theory provides an
index pairing between
 Breuer-Fredholm modules and $\KKEI$-theory
given by the Breuer-Fredholm index \cite{breuer1,breuer2,cprs2}.
As a characteristic of the Type II von Neumann algebra $\Nn$,
the Breuer-Fredholm index  now
 takes value in $\mathbb{R}$
as opposed to $\mathbb{Z}$ as in the Type I case.
In this paper, we will develop the even JLO character for Type II
non-commutative geometry.
 For completeness, we also
include the odd case, which was developed by Carey-Phillips in \cite{cp}.
 Similar to the Type I case,
 we show how to pass unbounded Breuer-Fredholm modules
 to bounded ones. 
Extending the argument for the Type I case,
 we show that this correspondence takes the class of
the (both even and odd) JLO character to that of the 
Connes  character in $\HE^\bullet(A)$, as defined in Type II case by 
Benameur-Fack \cite{type2index}.

The first section starts with background material
 on Breuer-Fredholm theory
 and the index pairing between $\KKEI$-homology and $\KKEI$-theory.
Then following \cite{type2index} we define the Connes-Chern character for $\KKEI$-homology and
the Chern character 
for $\KKEI$-theory, and show that
for  $p$-summable Breuer-Fredholm modules, these two characters intertwine the
index pairing with the pairing in entire cyclic (co)homology.
In Section 2, the JLO character for  $\theta$-summable unbounded Breuer-Fredholm modules
is defined. We study its homotopy invariance as an entire cyclic cohomology class 
by following along the lines of Getzler and Szenes \cite{getzlereven}
and show that it preserves the index pairing.
Section 3 connects the previous two sections by showing that a $p$-summable unbounded Breuer-Fredholm
module canonically gives rise to a $p$-summable Breuer-Fredholm module. We then proceed using
techniques from Connes and   Moscovici  \cite{chernreduct}
to show that
the JLO character for the $p$-summable unbounded Breuer-Fredholm module and the
Connes-Chern character for the  $p$-summable Breuer-Fredholm module define the same
entire cyclic cohomology class.
In the Appendix we recall some definitions and inequalities
needed for the discussion in our paper.

\section{Breuer-Fredholm modules and Connes-Chern character }

The section starts by stating the
 definition of Breuer-Fredholm modules from \cite{oddjloT2}.
With the notion of $(e,f)$-Fredholm from \cite{cprs3}, 
we proceed to develop a suitable Fredholm theory by following \cite{type2index}.
Entire cyclic (co)homology will be introduced, followed by a discussion of 
  the Chern character
\cite{getzlereven,getzlerodd} on $\KKEI$-theory and Connes-Chern character \cite{type2index} 
on $\KKEI$-homology.
The Section ends by showing that the characters intertwine the $\KKEI$-theoretical pairing
 given by the Fredholm index,
with the cohomological pairing.


\subsection{Breuer-Fredholm modules}

For a given semi-finite von Neumann algebra
$\Nn\subset B(\Hh)$ of bounded operators on a Hilbert space $\Hh$, with a faithful semi-finite normal trace
 $\tau$,
 we denote by
$\Kk$ the ideal of $\tau$-compact operators in $\Nn$. 
A $\tau$-compact operator is a (densely defined closed) operator affiliated with $\Nn$, such that its
generalized singular number $\mu_x(T)$ 
with respect to $\tau$ converges to $0$.
 The definitions and properties of 
$\Kk$ and $\mu_x(T)$
can be found
in the Appendix.

\begin{definition}
 \label{fredholmmodule}
An \textbf{\textit{odd} Breuer-Fredholm module} over a unital 
Banach $*$-algebra $A$ is a triple $\FREDMOD$ 
for which $\Nn$ is a (separable) semi-finite von Neumann algebra  with faithful semi-finite normal trace
 $\tau$,  
$\rho:A\rightarrow \Nn$ a continuous $*$-representation, 
 and $F\in \Nn$ an operator such that $F^2=1$ and $[F , \rho(a) ] \in \Kk$
for all $a\in A$.

If $\FREDMOD$ is equipped with a $\mathbb{Z}_2$ grading $\chi\in\Nn$ such that
all $\rho(a)$ are even and $F$ is odd,
then we call $\FREDMOD$ an \textbf{\textit{even} Breuer-Fredholm module}.

\end{definition}
If $\Nn=B(\Hh)$ and $\tau$ is the standard operator trace, we drop the prefix \textit{Breuer}.


As Fredholm modules are representatives of $\KKEI$-homology
classes in Kasparov's sense \cite{khomology},
 they are also referred to as
$\KKEI$-cycles.

Technically speaking, Breuer-Fredholm modules do \emph{not}
define $\KKEI$-homology classes in the usual sense, 
however one can still consider  its classes
given by the equivalence relations in $\KKEI$-homology.
 i.e.,
up to degenerate modules,
two such modules are equivalent if their Fredholm operators
are connected by a norm continuous homotopy of Fredholm operators
(in $\Nn$)   (see for example \cite{khomology} for a precise definition).
 We think of Breuer-Fredholm modules
as representatives of elements in some semi-finite or
Type II $\KKEI$-homology as \cite{cprs1,cprs2} did.
Whenever we write $[\FREDMOD]\in \KKEI^\bullet(A)$, we implicitly
mean that the $\KKEI$-homology is in the semi-finite sense.

Recall that a densely defined closed operator $T$ with 
polar decomposition
$T=U|T|$ is said to be affiliated with $\Nn$ if $U\in \Nn$ and also the spectral projections of $|T|$ lie 
in $\Nn$ (see Appendix). 
The only unbounded operators we are dealing with here are densely defined closed operators, 
hence the properties of an unbounded operator being densely defined and closed 
 are automatically assumed throughout this paper.
In particular,  when we speak of an operator $T$ affiliated with $\Nn$, we demand that
$T$ is densely defined and closed.
\begin{definition}
\label{fredholmoperator}
Given 
two projections
$e,f\in \Nn$, a (possibly unbounded) operator $T$
 affiliated with $\Nn$ is called
\textbf{$(e,f)$-Fredholm} if there is an operator
 $S\in\Nn$, 
 such that
\[e-e  S f  T  e \in {\mathcal{K}_{e\Nn e}} \hspace{1cm}
\mbox{  and  } \hspace{1cm}
f-f  T  e  S f \in {\mathcal{K}_{f\Nn f}} \mbox{ ,}\]
where $\mathcal{K}_{e\Nn e}$ denotes the set of $\tau$-compact operators in $e\Nn e$, likewise for 
$\mathcal{K}_{f\Nn f}$.
The operator $S$ is called an \textbf{$(e,f)$-parametrix} for $T$.
\end{definition}



\begin{example}
\label{ex:parametrixexamples}
\mbox{ }
\begin{itemize}
 \item 
Let  $\FREDMOD$ be a Breuer-Fredholm module.
If $u\in \Nn$ is a unitary, then $u$ is $(\frac{F+1}{2},\frac{F+1}{2})$-Fredholm with  $(\frac{F+1}{2},\frac{F+1}{2})$-parametrix $u^{-1}$. 
\item
Suppose that $\FREDMOD$ comes equipped with a $\mathbb{Z}_2$ 
grading $\chi$ and that the
  projection
$p\in\Nn$ is even with respect to $\chi$, 
then 
$F$ is $(p^+ ,p^-)$-Fredholm with $(p^+ ,p^-)$-parametrix $F$ again.
\end{itemize}\end{example}

The following Proposition can be found in \cite{type2index}.
We adopted it in the $(e,f)$-parametrix case.

\begin{proposition}
Let $T$ be a $(e,f)$-Fredholm operator, and $P_{\KER T}$ and $P_{\KER(T^*)}$ be the projections
onto the kernels of $T$ and  $T^*$ respectively.
Then $eP_{ \KER T}$ and $P_{\KER (T^*)}f$ have finite trace with respect to $\tau$.
\end{proposition}
\begin{proof}
Let $S$ be a $(e,f)$-parametrix of $T$ as in Definition~\ref{fredholmoperator}.
 We have  $(e-e  S f  T  e )P_{\KER T} = eP_{\KER T}$ and
$P_{\KER(T^*)}(f-f  T  e  S f ) =  P_{\KER(T^*)}f$, and 
 $eP_{\KER T}$ and $ P_{\KER(T^*)}f$ projections onto 
$\KER (T)\cap e(\Hh) = \KER (\left. fTe \right \lvert _{e(\Hh)} ) $ 
and $\KER (T^*)\cap f(\Hh)=\KER (\left. eT^* f \right \lvert _{f(\Hh)} )$ respectively.
By the ideal property of $\mathcal{K} _{e\Nn e}$, $eP_{\KER T}$ is a  $\tau$-compact projection.
As projections only have eigenvalue $\{0,1\}$,
$\tau$-compactness forces  the singular values of projections to have support in a bounded region, 
hence $\tau$ of any $\tau$-compact projection must be finite, and $\tau(eP_{\KER T})<\infty$.
Likewise,  $ \tau(P_{\KER(T^*)}f)<\infty$.
\end{proof}

\begin{definition}
The $(e,f)$-index $\IND(fTe)$ of an $(e,f)$-Fredholm operator $T$ is defined to be
\[
\IND(fTe):=\tau(eP_{\KER T}) - \tau(P_{\KER(T^*)}f)\mbox{ ,}
\]
where $P_{\KER T}$ and $P_{\KER (T^*)}$ are the projections onto the kernel of $T$ and $T^*$ respectively.
\end{definition}

Given an even Breuer-Fredholm module $\FREDMOD$ over $A$,
and a projection $p\in A$.
It follows from Example~\ref{ex:parametrixexamples}
that
$F$ is a $(\rho(p)^+,\rho(p)^-)$-Fredholm operator. Thus it has a well-defined $(\rho(p)^+,\rho(p)^-)$-index,
given by $\IND(\rho(p)^- F \rho(p)^+)$.

For a given odd  Breuer-Fredholm module $\FREDMOD$,
and a unitary $u\in A$, then $\rho(u)$ is a $(Q,Q)$-Fredholm operator, where $Q=\frac{F+1}{2}$.
Thus it has a well-defined $(Q,Q)$-index, given by $\IND(Q\rho(u) Q)$.

Since the function $\IND$ is locally constant \cite{cprs3},
the $(\rho(p)^+,\rho(p)^-)$-index descends to a pairing between the $\KKEI$-homology class 
$[\FREDMOD] \in \KKEI^0(A)$ and the $\KKEI$-theory class $[p]\in \KKEI_0(A)$.
Likewise,  the $(Q,Q)$-index descends to a paring between the classes
 $[\FREDMOD]\in \KKEI^1(A)$ and
 $[u]\in \KKEI_1(A)$.
 We extend the pairing to a pairing between  $\KKEI$-homology and $\KKEI$-theory of $A$ with the following
definition.

To simplify our notation, whenever we mention an element $a\in A$,
 we think of it as an operator $\rho(a)\in \Nn$ represented on $\Hh$, and will stop writing 
$\rho$. Similarly, when we have $a\in M_N(A)$, we think of it as an operator in $M_N(\Nn)$
represented on $\Hh^N=\Hh \otimes \mathbb{C}^N$ with the obvious
representation extended from $\rho$.

%
\begin{definition}[\cite{oddjloT2,cprs1,cprs2}]
\mbox{ }
\begin{enumerate}
\item 
Let $\FREDMOD$ be an even Breuer-Fredholm module  over $A$, representing the
$\KKEI$-homology class $[\FREDMOD]\in \KY ^0(A) $, and 
$p\in M_N(A)$ be a projection , representing the 
$\KKEI$-theory class $[p]\in \KY_0 (A)$.
We define the  \textbf{\textit{even} index pairing} to be:
\begin{eqnarray*}
\langle  [\FREDMOD], [p] \rangle &:=&
\IND(p^-(F\otimes 1_N) p^+)  \mbox{ .} \\
  \end{eqnarray*}
where $p^- (F\otimes 1_N) p^+$ is an operator from
 $p^+\Hh^N  $ to $p^- \Hh^N $.

\item 
Let $\FREDMOD$ be an odd Breuer-Fredholm module  over $A$,
 representing the
$\KKEI$-homology class $[\FREDMOD]\in \KY^1(A)$,
 and $u\in M_N(A)$ be a unitary, representing the 
$\KKEI$-theory class $[u]\in \KY_1(A) $.
We define the  \textbf{\textit{odd} index pairing} to be:
\[\langle [(\rho,\Nn,F) ], [u] \rangle:=
  \IND \left( Q  u  Q \right)   \mbox{ ,}
 \]
where $Q=\frac{F\otimes 1_N + 1}{2}$ is a projection in $M_N(\Nn)$, and
$QuQ$ is an operator from $Q\Hh^N$ to $Q\Hh^N$.
\end{enumerate}
\end{definition}

\subsection{Entire cyclic (co)homology}

Our goal is to construct characters from $\KKEI$-homology to another cohomology theory that
 intertwine the above $\KKEI$-theoretical pairing with the cohomological pairing.
The target space of both the Connes-Chern character and the JLO
character to be introduced next section is the entire cyclic (co)homology.

Entire cyclic homology is not as well-known as its cohomology counterpart.
We adopt the bicomplex construction from \cite{getzlerodd}
and use the entire growth control given in \cite{meyer}.
Under this definition, the homology theory is precisely
(pre-)dual to the cohomology counterpart \cite{getzlereven}
 in the sense that their pairing
produces a finite value.

If $\Bb$ is a topological unital algebra over $\mathbb{C}$, define
\[
 C_n(\Bb):=\Bb\hat{\otimes} (\Bb/\mathbb{C})^{\hat{\otimes} n}\mbox{ ,}
\]
where $\hat{\otimes}$ denotes the projective tensor product.
Denote the element $a_0\otimes\cdots\otimes a_n$ of  $C_n(\Bb)$ by $(a_0,\ldots,a_n)_n$, when
the context is clear we will omit the subscript $n$.
The operators $b: C_n(\Bb) \rightarrow  C_{n-1}(\Bb)$ and $B: C_n(\Bb)\rightarrow  C_{n+1}(\Bb)$ are given in terms of
simple tensors
by the formulas
\begin{eqnarray*}
 b(a_0,\ldots,a_n)_n&:=&\sum^{n-1}_{j=0} (-1)^{j} (a_0,\ldots,a_{j}a_{j+1},\ldots,a_n) _{n-1}
+ (-1)^{n}(a_n a_0,a_1,\ldots,a_{n-1})_{n-1} \mbox{ ,}\\
B(a_0,\ldots,a_n)_n
&:=&\sum^n_{j=0}(-1)^{nj}(1,a_j,\ldots,a_n,a_0,\ldots,a_{j-1})_{n+1} \mbox{ .}
\end{eqnarray*}

Simple calculations show that $b^2=0$, $B^2=0$, and $Bb+bB=0$.
Therefore $(b+B)^2=0$ and we get the following bicomplex:
\[
    {
\xygraph{
!{<0cm,0cm>;<3cm,0cm>:<0cm,3cm>::}
!{(1,5) }*+{\vdots}="15"
!{(2,5) }*+{\vdots}="25"
!{(3,5) }*+{\vdots}="35"
!{(4,5) }*+{\vdots}="45"
!{(0,4) }*+{\cdots}="04"
!{(1,4) }*+{C_3(\Bb)}="14"
!{(2,4) }*+{C_{2}(\Bb)}="24"
!{(3,4) }*+{C_{1}(\Bb)}="34"
!{(4,4) }*+{C_{0}(\Bb)}="44"
"14":"15"^{B}
"24":"25"^{B}
"34":"35"^{B}
"44":"45"^{B}
"04":"14"^{b}
"14":"24"^{b}
"24":"34"^{b}
"34":"44"^{b}
!{(0,3) }*+{\cdots}="03"
!{(1,3) }*+{C_2(\Bb)}="13"
!{(2,3) }*+{C_{1}(\Bb)}="23"
!{(3,3) }*+{C_{0}(\Bb)}="33"
"13":"14"^{B}
"23":"24"^{B}
"33":"34"^{B}
"03":"13"^{b}
"13":"23"^{b}
"23":"33"^{b}
!{(0,2) }*+{\cdots}="02"
!{(1,2) }*+{C_1(\Bb)}="12"
!{(2,2) }*+{C_{0}(\Bb)}="22"
"12":"13"^{B}
"22":"23"^{B}
"02":"12"^{b}
"12":"22"^{b}
!{(0,1) }*+{\cdots}="01"
!{(1,1) }*+{C_0(\Bb)}="11"
"11":"12"^{B}
"01":"11"^{b}
!{(2,1) }*+{}="a"
!{(4,3) }*+{}="b"
"a":"b"_{(b+B)}
}
}\mbox{ .}
\]
The space $C_\bullet(\Bb):=\prod_{n=0}^\infty C_{n}(\Bb)$  has a natural
$\mathbb{Z}_2$ grading given by 
 $C_+(\Bb)=\prod_{k=0}^\infty C_{2k}(\Bb)$ and $C_-(\Bb)=\prod_{k=0}^\infty C_{2k+1}(\Bb)$.
We get a chain complex $\left(C_\bullet(\Bb), b+B \right)$ with the odd boundary map
$b+B$. However, the homology of this chain complex is trivial
\cite{meyer}.
In order to make it nontrivial, we need to control the growth of a chain as $n$ varies.
The following definition is taken from \cite{getzlerodd,meyer}.
\begin{definition}
Define the space of \textbf{entire chains}
 \[ C_\bullet^\omega(\Bb):=\left\{ A_\bullet \in C_\bullet(\Bb):
 \sup_n\left( \left\lVert A_{n} \right\rVert _\pi \frac{\lambda^{n}}{\Gamma(\frac{n}{2})}\right)
 < \infty \mbox{ for some } \lambda >0 \right\}\]
 where $\left\lVert \cdot \right\rVert _\pi$ is the projective tensor norm.
 $\left(C_\bullet^\omega(\Bb),b+B\right)$ forms a subcomplex of \newline $\left(C_\bullet(\Bb), b+B \right)$.
The homology defined by $\left(C_\bullet^\omega(\Bb),b+B\right)$ is the \textbf{entire cyclic homology} 
of $\Bb$, denoted 
$\HE_\bullet(\Bb)=\HE_+(\Bb)\oplus \HE_-(\Bb)$. 
$\HE_\bullet(\Bb)$ is equipped with the obvious 
group structure inherited from the addition on $C_n(\Bb)$.
\end{definition}
We set $C^n(\Bb):=\operatorname{Hom}(C_n(\Bb),\mathbb{C})$ and let $(b+B):
C^\bullet(\Bb)\rightarrow C^\bullet(\Bb)$ 
be the transpose of the odd boundary map $(b+B):C_\bullet(\Bb)\rightarrow C_\bullet(\Bb)$
where  $C^\bullet(\Bb):=\prod_{n=0}^\infty C^{n}(\Bb)$,
 then
 we get a similar diagram as above with the arrows reversed.
The space $C^\bullet(\Bb)$ 
has a natural $\mathbb{Z}_2$ grading given by 
 $C^+(\Bb)=\prod_{k=0}^\infty C^{2k}(\Bb)$ and $C^-(\Bb)=\prod_{k=0}^\infty C^{2k+1}(\Bb)$.

 $\left(C^\bullet(\Bb), b+B \right)$ forms a cochain complex with the odd boundary map
$b+B$, which gives trivial cohomology \cite{meyer}.
\begin{definition}
\label{definitionentire}
Define the space of \textbf{entire cochains}
\[
 C^\bullet _\omega (\Bb):=\left\{ \phi _\bullet \in C^\bullet (\Bb): 
\sum_{n=0}^\infty  \Gamma(\frac{n}{2}) \left\lVert \phi_{n}\right\rVert
z^{n}
   \mbox{ is an entire function in $z$ }  \right\}
\]
where $\left\lVert\phi_n\right\rVert:=\sup\left\{|\phi_n(a_0,\ldots,a_n)|: \left\lVert a_j\right\rVert \leq 1  \mbox{ }\forall j\right\}$.
 $\left(C^\bullet_\omega(\Bb),b+B\right)$ forms a subcomplex of $\left(C^\bullet(\Bb), b+B \right)$.
The cohomology defined by $\left(C^\bullet _\omega(\Bb),b+B\right)$ is the \textbf{entire cyclic cohomology} of $\Bb$, denoted 
$\HE^\bullet(\Bb)=\HE^+(\Bb)\oplus \HE^-(\Bb)$. 
$\HE^\bullet(\Bb)$ is equipped with the obvious 
group structure inherited from the addition on $\operatorname{Hom}(C_n(\Bb),\mathbb{C})$.
\end{definition}


It is known that the de Rham homology (over $\mathbb{C}$) on a closed manifold $M$ is a
summand of the entire cyclic cohomology of the algebra $C^\infty(M)$.
They are expected to be equal, however it is not proved except
for the case when $M$ is one-dimensional \cite{boundary}.

Let
$\TR:C_n(M_N(A))\rightarrow C_n(A)$ 
be the  map defined by
\[
 \TR\left(m_0,m_1,\ldots,m_n \right):=\sum _{0\leq i_0,\ldots, i_n\leq N} 
\left( (m_0)_{i_0 i_1},(m_1)_{i_1 i_2},\ldots,
(m_n)_{i_n i_0}\right) \mbox{ ,}
\] 
where $\left(m_{k}\right) _{ij}$
 denotes the entries of the matrix $m_k$.
\begin{definition}
\mbox{ }
\begin{enumerate}
\item 
Let $p\in M_N(A)$ be a projection.
Define 
the \textbf{\textit{even} Chern character} $\CCHERN_+(p)\in C_+(A)$ of $p$ to be 
\[
 \CCHERN_+(p):=\sum_{k=0}^\infty  \CCHERN_{2k}(p)\mbox{ ,}
\]
where
\begin{eqnarray*}
\CCHERN_0(p)&:=&\TR(p)  \mbox{ ,}\\
\CCHERN_{2k}(p)&:=& (-1)^k \frac{(2k)!}{2\cdot k!} \TR(2p-1,p,\ldots,p)_{2k}   \mbox{ .}
\end{eqnarray*}
\item
Let  $u\in M_N(A)$ be a unitary.
Define the \textbf{\textit{odd} Chern character} $\CCHERN_-(u)\in C_-(A)$ of $u$ to be
\[
  \CCHERN_-(u):=\sum_{k=0}^\infty \CCHERN_{2k+1}(u) \mbox{ ,}
\]
where
\begin{eqnarray*}
 \CCHERN_{2k+1}(u)&:=& 
\frac{1}{\Gamma(\frac{1}{2})} (-1)^{k+1} k! \cdot \TR(u^{-1},u,\ldots,u^{-1},u)_{2k+1} \mbox{ .}
\end{eqnarray*}
\end{enumerate}
\end{definition} 

For convenience, we often  write $\TR\left( m_0,m_1,\ldots,m_n \right)$ 
simply as $(m_0,\ldots,m_n)$.

\begin{lemma}[\cite{getzlereven,getzlerodd}]
\label{1a}
\label{1b}
The Chern characters
$
 \CCHERN_+(p)  $  and  $\CCHERN_-(u)
$
are entire cyclic cycles in $\HE_+(A)$ and $\HE_-(A)$ respectively.\newline
That is, 
\[ \CCHERN_+(p) \in C^\omega_+(A)\mbox{ ,}\hspace{0.5cm}
 (b+B)\CCHERN_+(p)=0 \mbox{ ;}
\]
and
\[
 \CCHERN_-(u) \in C^\omega_-(A)\mbox{ ,}\hspace{0.5cm}
 (b+B)\CCHERN_-(u)=0 \mbox{ .}
\]
Furthermore, the homology classes $[\CCHERN_+(p) ]$
and $[\CCHERN_-(u)] $ depend 
only on the $\KKEI$-theory classes of $[p]\in  \KKEI _0 (A)$
and $[u] \in \KKEI_1(A)$ respectively.
 \end{lemma}

As a result of Lemmas~\ref{1a},
 the Chern character $\CCHERN_\bullet$ descends to 
a map from $\KY_\bullet(A)$ to $\HE_\bullet(A)$. It is easy to 
see that $\CCHERN_\bullet$ respects group additions, hence it is
a group homomorphism.

\subsection{Connes-Chern character}

The Connes-Chern character is a cohomological 
Chern character due to Connes that
  assigns to a Breuer-Fredholm module a cocycle in entire cyclic cohomology.
However, not every Breuer-Fredholm module lies inside the domain of the Connes-Chern character.
To characterize those that are within the domain, we need the following summability condition.

 For $0<p<\infty$, let $\Ll^p$ be the set of $p$-summable operators in $\Nn$.
That is, an operator $T$ is in $\Ll^p$ if $T\in\Nn$ and 
its $p$-norm $\lVert T \rVert_p$ with respect to $\tau$ is finite.
  More details 
can be found
in the Appendix.

The Connes character for the Type II setting
first appeared in the work of Benameur and Fack in \cite{type2index}.
Let $A$ be a Banach $*$-algebra.

\begin{definition}
A Breuer-Fredholm module $\FREDMOD$ over $A$ is called 
$p$-summable for $[F,\rho(a)]$ is $p$-summable for all $a\in A$ (see 
Definition~\ref{summablecompactmeasurable}).
\end{definition}

\begin{definition} \mbox{ }
\begin{enumerate}
\item 
Recall that $\chi$ is the grading operator that anti-commutes with $F$.
Define the  \textbf{\textit{even} Connes character} $\ch^n(F)$ of
an even $p$-summable Breuer-Fredholm module $\FREDMOD$ to be
the linear functional   on $C_n(A)$ given by
\[
 \left( \ch^n(F),(a_0,\ldots,a_n)_n \right) := 
\frac{\Gamma(\frac{n}{2}+1)}{2\cdot n!} \tau (\chi F[F,a_0][F,a_1]\cdots[F,a_n])\mbox{ ,}
\]
where $n$ is an  even integer greater than $p$ and $\left( \cdot, \cdot \right)$ denotes the pairing between cochains and chains. 
\item
Define the \textbf{\textit{odd} Connes character} $\ch^n(F)$ of
an odd $p$-summable Breuer-Fredholm module $\FREDMOD $ to be
the linear functional   on $C_n(A)$ given by
\[
 \left( \ch^n(F),(a_0,\ldots,a_n)_n \right) := 
\frac{\Gamma(\frac{n}{2}+1)}{2\cdot n!} \tau ( F[F,a_0][F,a_1]\cdots[F,a_n])\mbox{ ,}
\]
where $n$ is an odd integer greater than $p$ and $\left( \cdot, \cdot \right)$ 
denotes the pairing between cochains and chains. 
\end{enumerate}
\end{definition}

\begin{theorem}
For $n>p$,
the even/odd Connes character $\ch^n(F)$ of an even/odd $p$-summable Breuer-Fredholm module $\FREDMOD$
defines an entire cyclic cocycle, its cohomology class is independent of $n$ with the same parity.
\end{theorem}
\begin{proof}
It is clear that $\ch^n(F)$ entire.
Since $B\ch^n(F)=0$,   we only need to show that $b\ch^n(F)=0$, and 
that $\ch^n(F)-\ch^{n+2}(F)$ is exact.
A short computation shows that 
\[\ch^n(F)=B \psi ^{n+1}(F) \hspace{1cm}\text{  and  } \hspace{1cm}
 -\ch^{n+2}(F)= b \psi ^{n+1}(F)\] where 
the entire cochain $\psi ^{n+1}(F)$ is given by
\begin{eqnarray*}
 \left( \psi ^{n+1}(F), (a_0,\ldots,a_{n+1})_{n+1} \right)
 :=\frac{ \Gamma{(\frac{n}{2}+2)} }
{(n+2)!}  \times 
\left\{
 \begin{array}{ll}
 
\tau (\chi a_0F[F,a_1][F,a_2]\cdots[F,a_{n+1}] ) 
& n \mbox{ even}\\
 \tau (a_0F[F,a_1][F,a_2]\cdots[F,a_{n+1}] ) 
& n \mbox{ odd}
\end{array} 
\right.
\mbox{ .}
 \end{eqnarray*}

Thus,
 \[  b\ch^{n+2}(F)=bb( -\psi^{n+1}(F))=0 \]  and  \[ 
 \ch^n(F)-\ch^{n+2}(F) = (b+B)  \psi^{n+1}(F) \mbox{ ,}
 \]
which completes the proof.
\end{proof}

Suppose that $F_t$ is a norm continuous family of 
Fredholm operators parametrized by $t$
so that $(\rho,\Nn,F_t)$ 
defines 1-parameter family of $p$-summable Breuer-Fredholm
modules. 
\begin{theorem}
\label{conneschartransgression}
The entire cyclic cohomology class defined
by the Connes character $\ch^n(F_t)$ is independent
of $t$.
More explicitly 
 \[\frac{d}{dt}\ch^n(F_t)=(b+B) \iota( \dot{F_t} )
\ch^{n-1}(F_t) \mbox{ ,}
\]
where  the entire cochain $ \iota( \dot{F_t} )
\ch^{n-1}(F_t)$ is given by
\begin{eqnarray*}
\lefteqn{\left( \iota( \dot{F_t} )
\ch^{n-1}(F_t), (b_0,\ldots,b_{n-1})_{n-1} \right)
}\\ &&:=
\frac{\Gamma(\frac{n}{2}+1)}{  n!} 
\sum _{k=0} ^{n-1}  (-1)^{k+1}
\tau (\chi F_t [F_t, b_0 ] \cdots 
[F_t,b_k]\dot{F}_t\cdots [F_t,b_{n-1}]) 
\end{eqnarray*}
with the convention that $\chi=1$ in the odd case.
\end{theorem}
\begin{proof}
$\dot{F_t}$ being bounded implies that $\iota( \dot{F_t} )
\ch^{n-1}(F_t)$ is entire.
It is clear that
 \[ B\iota( \dot{F_t} )\ch^{n-1}(F_t)=0  \mbox{ ,}
\]
 so we only need to show that 
 \[\frac{d}{dt}\ch^n(F_t)=b  \iota( \dot{F_t} )
\ch^{n-1}(F_t) \mbox{ .}
\]
Let $T_k \in C^{n-1}(A)$  be defined by the equation
\begin{eqnarray*}
\left( T_k,(b_0,\ldots,b_{n-1})_{n-1} \right) &:=&
(-1)^{k} 
\tau (\chi F_t [F_t, b_0 ] \cdots 
\dot{F}_t[F_t,b_k]\cdots [F_t,b_{n-1}]) \mbox{ ,}
\end{eqnarray*}
then
\begin{eqnarray*}
\left(b T_k,(a_0,\ldots,a_{n}) \right)&=& (-1)^{k}
\Biggl(
\sum_{j=0} ^{k-1}
(-1)^j
\tau(\chi F_t  \cdots
[F_t,a_j a_{j+1} ] 
\cdots
  [F_t,a_k] \dot{F}_t\cdots
[F_t,a_n] ) 
 \\ && \hspace{1cm}
+
\sum_{j=k} ^{n-1}
(-1)^j
\tau(\chi F_t \cdots \dot{F}_t 
[F_t,a_k]
\cdots
[F_t,a_j a_{j+1} ] 
\cdots
[F_t,a_n] ) 
\\&& \hspace{1cm}
+ 
(-1)^n
\tau(\chi F_t 
[F_t,a_n a_0]
\cdots \dot{F}_t 
[F_t,a_k] 
\cdots
[F_t,a_{n-1}] ) \Biggr) \mbox{ .}
\end{eqnarray*}
By expanding the term $[F_t,a_j a_{j+1}]=
a_j[F_t, a_{j+1}] +[F_t,a_j ] a_{j+1}$, we see that
the above becomes a telescope sum and most of the terms 
cancel. Using the identity 
\[0=\frac{d}{dt}\left( F_t\cdot F_t\right)
= \dot{F_t}\cdot F_t + F_t\cdot \dot{F_t}\]
 we  simplify further and obtain
\begin{eqnarray*}
\left(b T_k,(a_0,\ldots,a_{n})_n \right)
&=&
(-1)^{k}
\left(
\tau(\chi F_t a_0 [F_t,a_1] \cdots [F_t,a_k] \dot{F}_t
\cdots[F_t,a_n] ) \right.
\\ &&+
(-1)^{k} \tau(\chi F_t \cdots [\dot{F}_t,a_k] \cdots
[F_t,a_n] )
\\&&+ \left.
(-1)^n \tau(\chi F_t  [F_t,a_n]a_0
  \cdots \dot{F}_t [F_t,a_k]
\cdots [F_t,a_{n-1}] ) \right)\\
&=&
\label{desiredcommutator}
 \tau(\chi F_t \cdots [\dot{F}_t,a_k] \cdots
[F_t,a_n] ) \\&&
\label{telescopeA}
(-1)^{k} 
\tau(\chi F_t a_0 [F_t,a_1] \cdots [F_t,a_k] \dot{F}_t
\cdots[F_t,a_n] ) 
\\&&
\label{telescopeB}
 - 
(-1)^{k-1} \tau(\chi F_t a_0 [F_t,a_1] \cdots [F_t,a_{k-1}] \dot{F}_t
\cdots[F_t,a_n] )  \mbox{ .}
\end{eqnarray*}
By construction, 
\[
\iota( \dot{F_t} )
\ch^{n-1}(F_t) = \frac{\Gamma(\frac{n}{2}+1)}{  n!}  
\sum_{k=1}^n T_k \mbox{ ,}
\]
and hence
\begin{eqnarray*}
\left( b \iota( \dot{F_t} )
\ch^{n-1}(F_t) , (a_0,\ldots,a_n)_n \right)
&=& \frac{\Gamma(\frac{n}{2}+1)}{  n!}
\left(
\sum_{k=1}^n  \tau(\chi F_t [F_t,a_0]\cdots [\dot{F}_t,a_k] \cdots
[F_t,a_n] ) \right. \\
&& \hspace{2cm} - \tau(\chi F_t a_0 \dot{F}_t [F_t,a_1] \cdots
[F_t,a_n] ) \\ && \hspace{2cm} \left.
+ (-1)^{n}
\tau(\chi F_t a_0 [F_t,a_1]\ldots [F_t,a_n]\dot{F_t}) \right)
\\
&=& \frac{\Gamma(\frac{n}{2}+1)}{  n!} \left(
\sum_{k=0}^n  \tau(\chi F_t [F_t,a_0]\cdots [\dot{F}_t,a_k] \cdots
[F_t,a_n] ) \right. \\
&& \hspace{2cm} + \tau(\chi  \dot{F}_t F_t a_0 [F_t,a_1] \cdots
[F_t,a_n] ) \\ && \hspace{2cm} \left.
-
\tau(\chi  a_0 [F_t,a_1]\ldots [F_t,a_n]\dot{F_t}F_t) \right)
\\
&=&
\frac{\Gamma(\frac{n}{2}+1)}{  n!} \left(
 \sum_{k=0}^n  \tau(\chi F_t [F_t,a_0]\cdots [\dot{F}_t,a_k] \cdots
[F_t,a_n] ) \right. \\
&& \hspace{2cm} \left. + \tau(\chi \dot{F}_t [ F_t , a_0]  \cdots
[F_t,a_n] ) \right)\\
&=&
\left( \frac{d}{dt}\ch^n(F_t), (a_0,\ldots,a_n)_n \right)
\mbox{ .}
\end{eqnarray*}
The proof is complete.
\end{proof}
In fact, the way we obtain the transgression formula
in Theorem~\ref{conneschartransgression} is by taking
limits of the transgression formula in Theorem~\ref{cchisexact} below.



\begin{proposition}[\cite{type2index}]
\label{index}
Suppose that $T$ is a $(e,f)$-Fredholm operator with parametrix $S$ such that
\[
 e-eSfTe \in \mathcal{L}_{e\Nn e}^{p/2} 
\hspace{1cm}\mbox{  and }\hspace{1cm}
 f-fTeSf \in \mathcal{L}_{f\Nn f}^{p/2} 
\mbox{ ,}
\]
where $\mathcal{L}_{e\Nn e}^{p/2} $ denote the set of $\frac{p}{2}$-summable operators in $e\Nn e$, likewise
for $\mathcal{L}_{f\Nn f}^{p/2} $.
Then
\[
 \IND(fTe)=\tau\left( (e-eSfTe)^m \right) - \tau \left( (f-fTeSf)^m \right) 
\]
for $2m>p$.


\end{proposition}

The following theorem shows that
the characters $\ch^n$ and $\CCHERN_\bullet$ intertwine the
$\KKEI$-theoretical pairing with the (co)homological pairing of
entire cyclic (co)homology.

\begin{theorem}[\cite{type2index}] \mbox{ }
\begin{enumerate}
\item 
Let $\FREDMOD$ be an even $\operatorname{p}$-summable Breuer-Fredholm module and
$p\in M_N(A)$ be a projection, then for $n>\operatorname{p}$ even 
\begin{eqnarray*}
 \langle [\FREDMOD],[p]  \rangle &=& \left([ \ch^n(F)],[\CCHERN_+(p)] \right)
 \mbox{ .}
\end{eqnarray*}
\item
Let $\FREDMOD$ be an odd $\operatorname{p}$-summable Breuer-Fredholm module and
$u\in M_N(A)$ be a unitary, then for $n>\operatorname{p}$ odd
\begin{eqnarray*}
\langle[ \FREDMOD ],[u] \rangle &=& \left( [\ch^n(F)],[\CCHERN_-(u)] \right) \mbox{ .}
\end{eqnarray*}
%
\end{enumerate}
\end{theorem}

\section{Unbounded Breuer-Fredholm modules and JLO character }
This section repeats the language in Section 1 for  unbounded Breuer-Fredholm modules.
It starts with the definition of 
unbounded Breuer-Fredholm modules from \cite{oddjloT2}
and its pairing with $\KKEI$-theory.
The JLO character is defined
 and a proof of its homotopy invariance is shown according to \cite{getzlereven}.
The section concludes by showing that the JLO character computes the index.

Much of the work in this section is taken directly from \cite{getzlereven} with minor modifications.
Nonetheless,
we give full details to illustrate the changes made in this Type II setting.

\subsection{Unbounded Breuer-Fredholm modules}
\begin{definition}
\label{spectraltriple}
An \textbf{\textit{odd} unbounded Breuer-Fredholm module}
 over a unital Banach $*$-algebra  $A$ 
is a triple 
$\KCYCLE$ for which 
$\Nn$ is a (separable) semi-finite von Neumann algebra in $B(\Hh)$ with a faithful semi-finite normal trace $\tau$,
 $\rho:A\rightarrow \Nn$ a continuous $*$-representation, 
  and $\Dd$  is an unbounded self-adjoint operator on $\Hh$ such that
\begin{enumerate}
\item 
$\Dd$ is affiliated with $\Nn$,
\item 
For all $a\in A$, the commutator $[\Dd , \rho(a) ]$
 extends to an operator in $\Nn$
 and 
there is a constant $C$ such that
$ \left\lVert [\Dd,\rho(a)]\right\rVert \leq C 
\left\lVert a\right\rVert$.
\item 
$(1+\Dd^2)^{-1/2}   \in \Kk$.
 \end{enumerate}
If $\KCYCLE$ is equipped with a $\mathbb{Z}_2$ grading $\chi\in\Nn$ such that
all $\rho(a)$ are even and $\Dd$ is odd,
then we call $\KCYCLE$
 an \textbf{\textit{even}  unbounded Breuer-Fredholm module}.
\end{definition}
If $\Nn=B(\Hh)$ and $\tau$ is the standard operator trace, we drop the prefix \textit{Breuer}.

To avoid confusion, we will sometimes refer to the Breuer-Fredholm module from
Definition~\ref{fredholmmodule} as \textit{\textbf{bounded}}.
Similar to its bounded counterpart, an unbounded Fredholm module
is sometimes called an \emph{unbounded} $\KKEI$-cycle.

The term (semi-finite) spectral triple
seems to be popular among physicists.
It is a convenient term for the package consisting of 
the algebra $A$ and an unbounded (Breuer-)Fredholm module.
In this thesis, our algebra $A$ is always fixed and we
 view the JLO character and Connes character
as maps from $\KKEI$-homology classes to some cohomology classes
that respect group additions. Hence,  the term unbounded
Breuer-Fredholm module is more convenient and
 suitable in our settings.

An example of an unbounded Breuer-Fredholm module is given by
 the semi-finite spectral triple over a space of $G$-connections 
due to Aastrup, Grimstrup, and Nest.
Similar to the Breuer-Fredholm module case,
 we think of an element $a\in A$ as an operator $\rho(a)\in \Nn$ represented on $\Hh$, and will stop writing 
$\rho$.

In Sections~\ref{unbounded2bounded},
 we will explain in details
 how we would associate a bounded Breuer-Fredholm
module to an unbounded one.


\begin{definition}
\mbox{ }
\begin{enumerate}

\item 

For a given even  
 unbounded Breuer-Fredholm module $\KCYCLE$ over $A$,
 define its pairing with the
even   $\KKEI$-theory $\KY_0(A)$   of 
$A$ given by the index:
\[
\langle [ \KCYCLE ], [p] \rangle:=
 \IND\left(p^- (\Dd \otimes 1_N)p^+ \right)\]   
for a projection $p\in M_N(A)$
 representing the class $[p]\in \KY_0(A) $, where 
\[
 p^-  (\Dd \otimes 1_N)p^+:
p^+ \Hh^N \longrightarrow p^-\Hh^N \mbox{ .} \]

\item 
For a given odd
 unbounded Breuer-Fredholm module $\KCYCLE$ over $A$,
 define its pairing with the
odd $\KKEI$-theory $\KY_1(A)$  of 
$A$ given by the spectral flow:
\[\langle [\KCYCLE ], [u] \rangle:=
\operatorname{sf}\left(\Dd\otimes 1_N,u(\Dd\otimes 1_N) u^{-1}\right)\] 
  for a unitary $u\in M_N(A)$ representing the class $[u]\in \KY_1(A)$,
where \newline
 $\operatorname{sf}\left(\Dd\otimes 1_N,u(\Dd\otimes 1_N) u^{-1}\right) $ is the spectral flow from 
$(\Dd\otimes 1_N ) \left(1+(\Dd\otimes 1_N) ^2 \right)^{-\frac{1}{2}} $
 to \newline $( u(\Dd\otimes 1_N) u ^{-1})
\left(1 + (u(\Dd\otimes 1_N) u^{-1})^2 \right)^{-\frac{1}{2}}$
 defined in  \cite{oddjloT2}.
\end{enumerate}
\end{definition}



\subsection{JLO character}

The JLO character
is a cohomological Chern character 
due to Jaffe, Lesniewski, and Osterwalder that assigns cocycles in entire cyclic cohomology to 
 unbounded Breuer-Fredholm modules satisfying an appropriate summability condition. We begin by defining the summability conditions of 
main concern.


\begin{definition}
 An unbounded Breuer-Fredholm module $\KCYCLE$ over 
$A$ is:
\begin{itemize}
\item[(a)]
\textbf{$p$-summable} if $\tau\left((1+\Dd^2)^{-p/2}\right) <\infty$ ;
\item[(b)]
\textbf{$\theta$-summable} if $
 \tau(e^{-t\Dd^2}) < \infty $ for all $t>0$;
\item[(c)]
\textbf{weakly $\theta$-summable} 
if $ \tau(e^{-t\Dd^2}) < \infty $ for some $0<t<1$.

\end{itemize}
\end{definition}

Observe that 
$p$-summability implies $\theta$-summability,
which in turn implies weak $\theta$-summability.


\begin{example}
 Let $\Gamma \hookrightarrow \tilde{M} \hookrightarrow M$ be a Galois cover of a compact 
$p$-dimensional manifold $M$. Let $\Dd$ be the $\Gamma$ cover of a generalized Dirac operator on $M$.
Consider the von Neumann algebra $\Nn$ of bounded $\Gamma$-invariant operators defined by Atiyah,
 with its natural trace $\TR_\Gamma$. $\Hh$ the Hilbert space $\Nn$ represents on, then 
$(\rho,\Nn,\Dd)$ is a $p$-summable unbounded Breuer-Fredholm module over $C^\infty(M)$ with $\rho$ given by point-wise multiplication \cite{type2index}.
\end{example}

\begin{example}
The unbounded Breuer-Fredholm module 
given by Aastrup-Grimstrup-Nest's noncommutative space of 
connections  is 
weakly $\theta$-summable if the sequence 
$\{a_j\}$ in its definition diverges sufficiently fast \cite{agn1}.
\end{example}

The following Lemma was proved in \cite{chernreduct} in the Type I case.
\begin{lemma}
\label{ptheta}
 If $\KCYCLE$ is $p$-summable for any finite $p$, then
 it is also $\theta$-summable, and $\tau(e^{-t\Dd^2})=O(t^{-p/2})$ as $t\searrow 0$.
\end{lemma}
\begin{proof}
 We can write $e^{-t\Dd^2}=(1+\Dd^2)^{p/2}e^{-t\Dd^2}(1+\Dd^2)^{-p/2}$ with
$\tau((1+\Dd^2)^{-p/2})<\infty$ by hypothesis, and $(1+\Dd^2)^{p/2}e^{-t\Dd^2}$  bounded by
$\left\lVert (1+x^2)^{p/2}e^{-tx^2}\right\rVert _\infty=\left(\frac{p}{2e}\right)^{p/2}t^{-p/2}e^t$ by functional calculus. Hence
as a consequence of Proposition~\ref{muproperties} and Proposition~\ref{taupositive}, we have
\[
 \tau(e^{-t\Dd^2} )\leq \left(\frac{p}{2e}\right)^{p/2}t^{-p/2}e^t\tau((1+\Dd^2)^{-p/2})\mbox{ ,}
\]
which proves the lemma.
\end{proof}


To make the JLO character and other useful formulas easier to write down,
 we will define the JLO character in two steps.
We start with the following definition.


Let $\Delta_n:=\{(t_1,\ldots,t_n)\in \mathbb{R}^n ; 0\leq t_1 \leq \cdots \leq t_n \leq 1\}$ be the
standard $n$-simplex and $d^nt=dt_1\cdots dt_n$ is the standard Lesbeque measure on $\Delta_n$ with
volume $\frac{1}{n!}$.
\begin{definition} 
Let $\KCYCLE$ be a weakly $\theta$-summable unbounded Breuer-Fredholm module over $A$.
Given $F_0,\ldots,F_n$ operators affiliated with $\Nn$, define 
 \[
  \langle F_0,F_1,\ldots,F_n \rangle ^n _\Dd :=\int_{\Delta_n}\tau 
\left( \chi 
F_0 e^{-t_1 \mathcal{D}^2 }
F_1 e^{-(t_2-t_1) \Dd^2 }
\ldots
F_n e^{-(1-t_n) \Dd^2 } 
\right)d^nt \mbox{ ,}
 \]
where $\chi=1$ when $\Dd$ is even.
\end{definition}

Let $T$ be an operator affiliated with $\Nn$,
denote by $\lvert T \rvert_\chi$  
the degree of $T$ with respect to $\chi$.
Any operators that we will consider will be \emph{either} even or odd.
From here and on, the commutator $[\mbox{ } ,\mbox{ } ]$ is always graded with respect to $\chi$.

\begin{lemma}
Let  $F_0,\ldots,F_n$ be operators affiliated with $\Nn$
that are either even or odd,
 then
\label{misc}
\begin{enumerate}
 \item 
\[
 \left \langle F_0,\ldots,F_n \right\rangle ^n _\Dd= (-1)^{(|F_0|_\chi+\cdots+|F_{j-1}|_\chi)(|F_j|_\chi+\cdots+|F_n|_\chi)}
\left\langle F_j,\ldots,F_n,F_0,\ldots,F_{j-1} \right\rangle _\Dd^n
\mbox{ ;}
\]
\item 
\[
 \left\langle F_0,\ldots,F_n \right\rangle ^n_\Dd
=\sum^n_{j=0}\left\langle F_0,\ldots,1,F_{j},
\ldots,F_n \right\rangle ^{n+1}_\Dd
\mbox{ ;}
\]
\item 
\[
 \sum^n_{j=0}(-1)^{|F_0|_\chi+\cdots+|F_{j-1}|_\chi}
\left \langle F_0,\ldots,[\Dd,F_j],\ldots,F_n 
\right \rangle ^n _\Dd =0\mbox{ ;}
\]
\item 
\begin{eqnarray*}
\left \langle F_0,\ldots,[\Dd^2,F_j],\ldots,F_n\right\rangle ^n _\Dd &=&
\left\langle F_0,\ldots,F_{j-1}F_j,F_{j+1},\ldots,F_n \right\rangle ^{n-1}_\Dd
\\ && \hspace{1.5cm}-
\left\langle F_0,\ldots,F_{j-1},F_j F_{j+1},\ldots,F_n \right\rangle ^{n-1}_\Dd \mbox{ .}
\end{eqnarray*}
\end{enumerate}
\end{lemma}
\begin{proof}\mbox{ }

 \begin{enumerate}
\item 
The statement follows from $\tau(\chi [X,Y])=0$ for $X,Y$ operators
affiliated with $\Nn$.
\item 
The left hand side can be regarded as $\int_0^1 \left\langle F_0,\ldots,F_n \right\rangle ^n_\Dd du$ by introducing
a trivial extra integration; the polyhedron $\Delta_n\times [0,1]$ can be subdivided by the inequalities
$t_j\leq u \leq t_{j+1}$ into $n+1$ simplices, each of which is a copy of $\Delta_{n+1}$;
integration over these simplices yield the terms on the right hand side.
\item 
By observing the Leibniz property of $[\Dd,\cdot]$ and
\[
 0=\tau\left(\chi[\Dd,F_0e^{-t_1\Dd^2}F_1e^{-(t_2-t_1)\Dd^2}\cdots F_n e^{-(1-t_n)\Dd^2} ] \right)\mbox{ ,}
\]
equality follows.
\item 
We first prove that\[
                    0=[e^{-\Dd^2},X]+\int_0^1 e^{-s\Dd^2}[\Dd^2,X]e^{-(1-s)\Dd^2}ds \mbox{ .}
                   \]
It comes from
\begin{eqnarray*}
 [e^{-\Dd^2},X]&=& 
\left.  e^{-s\Dd^2}X e^{-(1-s)\Dd^2} \right \rvert ^1_0 = \int_0^1 \frac{d }{ds}( e^{-s\Dd^2} X e^{-(1-s)\Dd^2} )ds\\
&=&\int_0^1 e^{-s\Dd^2} (-\Dd^2)X e^{-(1-s)\Dd^2} + e^{-s\Dd^2} X \Dd^2 e^{-(1-s)\Dd^2} ds \\
&=& - \int_0^1 e^{-s\Dd^2}[\Dd^2,X]e^{-(1-s)\Dd^2}ds \mbox{ .}
\end{eqnarray*}
Replacing $\Dd^2$ by $(t_{j+1}-t_j)\Dd^2$ and using the substitution $u=(t_{j+1}-t_j)s+t_j$, we obtain
\[                                                                                               
         0=[e^{-(t_{j+1}-t_j)\Dd^2},X]+\int_{t_j}^{t_{j+1}}
 e^{-(t_{j+1}-u)\Dd^2}[\Dd^2,X]e^{-(u-t_j)\Dd^2}du \mbox{ .}                                                        
                                        \]
Inserting this into the definition of $\langle F_0,\ldots,[\Dd^2,F_j],\ldots,F_n \rangle^n_\Dd$ gives the formula.
 \end{enumerate}
\end{proof}

\begin{definition}
\label{JLO}
\mbox{ }
\begin{enumerate}
\item 
 The  \textbf{odd JLO character}
 $\Ch^-(\Dd)\in C^-( A)$ 
of a weakly $\theta$-summable \textit{odd} 
unbounded Breuer-Fredholm module
$\KCYCLE$ is defined to be

\[
 \Ch^-(\Dd):=\sum_{k=0}^\infty \Ch^{2k+1}(\Dd)\mbox{ ,}
\]
\item
The \textbf{even JLO character} $\Ch^+(\Dd)\in C^+( A)$
of a weakly $\theta$-summable \textit{even} 
unbounded Breuer-Fredholm module
 $\KCYCLE$ is defined to be
\[
 \Ch^+(\Dd):=\sum_{k=0}^\infty \Ch^{2k}(\Dd)\mbox{ ,}
\]
where
 \[
  \left(\Ch^{n}(\Dd),(a_0,\ldots,a_{n})_n\right):=
\left\langle a_0,[\Dd,a_1],\ldots,[\Dd,a_{n}]\right\rangle^n _{\Dd} \mbox{ .}
 \]
\end{enumerate}
\end{definition}


\begin{theorem}
\label{jloisacocycle}
The JLO character $\Ch^\bullet(\Dd)$ is an entire cyclic cocycle in $\HE^\bullet( A)$.\newline
More specifically,
\[
 \Ch^\bullet(\Dd)\in C^\bullet _\omega( A)\hspace{0.3cm} \mbox{ and }\hspace{0.3cm}
(b+B)\Ch^\bullet(\Dd)=0\mbox{ .}
\]

\end{theorem}

The following norm estimate will show that  $\Ch^\bullet(\Dd)$  is entire.

Whenever we have an operator affiliated with $\Nn$, we demand that it is
either even or odd with respect to $\chi$.

\begin{lemma}
\label{getzler}
Let $\KCYCLE$ be a weakly $\theta$-summable
unbounded Breuer-Fredholm module over $A$.
If $F_j$ and $R_j$ are operators in $\Nn$ for $j=0,\ldots,n$, and at most $k$ of the operators $F_j$ are non-zero,
then for $\varepsilon \in [0, 1)$,
 \[
  \left\lvert \left\langle F_0 |\Dd|^{1+\varepsilon} +R_0,\ldots,F_n |\Dd|^{1+\varepsilon}  +R_n \right\rangle ^n _\Dd \right \rvert
\leq
\left(\frac{2}{(1-\varepsilon)\delta e }\right)^k
 \frac{\tau \left( e^{ -(1-\delta)\Dd^2  } \right)}{(n-k)!}
\prod ^n _{j=0}
\left(\left\lVert F_j\right\rVert+\left\lVert R_j\right\rVert  \right)
 \]
where $0<\delta<\frac{1}{2e}$.

\end{lemma}
For the purpose of future applications, Lemma~\ref{getzler} is slightly strengthened
 from the one in \cite{getzlereven}.
The proof in \cite{getzlereven} carries through to our setting with minor modications.
\begin{proof}
 From the generalized H\"{o}lder's inequality, Theorem~\ref{holder}(1), the following estimate holds:
\[
 |\tau(\chi T_0 \ldots T_n)| \leq \tau(|\chi T_0 \ldots T_n|) = \left\lVert\chi T_0 \ldots T_n\right\rVert_1
\leq\left\lVert T_0\right\rVert_{s_0^{-1}}\ldots\left\lVert T_n\right\rVert_{s_n^{-1}}
\]
if $s_0+\cdots+s_n=1$. Therefore,
\begin{eqnarray*}
\left\lvert \langle F_0|\Dd|^{1+\varepsilon}+R_0,\ldots,F_n|\Dd|^{1+\varepsilon}
+R_n\rangle \right\rvert
 \hspace{6cm} \\ \leq
\int_{\Delta_n}
\left\lVert(F_0|\Dd|^{1+\varepsilon}+R_0)e^{-s_0\Dd^2}\right\rVert_{s_0^{-1}} 
\cdots 
\left\lVert(F_n|\Dd|^{1+\varepsilon}+R_n)e^{-s_n\Dd^2}\right\rVert_{s_n^{-1}} d^ns \mbox{ .}
\end{eqnarray*}
For each $\left\lVert(F|\Dd|^{1+\varepsilon}+R)e^{-s\Dd^2}\right\rVert_{s^{-1}}$, observe that by using Proposition~\ref{pnormproduct}
and functional calculus
\begin{eqnarray*}
 \left\lVert F|\Dd|^{1+\varepsilon} 
e^{-s\Dd^2}\right\rVert_{s^{-1}} & \leq&
\left\lVert F\right\rVert \cdot \left\lVert |\Dd|^{1+\varepsilon} e^{-\delta s \Dd^2}\right\rVert\cdot
\left\lVert e^ {-s(1-\delta) \Dd^2 }\right\rVert
_{s^{-1}} \\
& \leq &
\left\lVert F\right\rVert \cdot 
\sup_{x \in \mathbb{R}} \left( |x|^{1+\varepsilon}e^{-\delta s x^2} \right)
\cdot
\left\lVert e^ {-s  (1-\delta) \Dd^2}\right\rVert
_{s^{-1}}
\end{eqnarray*}
and that
\begin{eqnarray*}
 \left\lVert Re^{-s\Dd^2}\right\rVert_{s^{-1}} &\leq&
\left\lVert R\right\rVert\cdot \left\lVert e^{-s\delta\Dd^2} \right\rVert \cdot \left\lVert e^ {-s(1-\delta)   \Dd^2}\right\rVert
_{s^{-1}}
\\&\leq&
\left\lVert R\right\rVert\cdot  
\sup_{x \in \mathbb{R}} \left(e^{-s\delta x^2} \right)
 \cdot \left\lVert e^ {-s (1-\delta)  \Dd^2}\right\rVert
_{s^{-1}} \mbox{ .}
\end{eqnarray*}
Since the function $|x|^{1+\varepsilon} e^{-\delta s x^2}$ is bounded by
 $\left(\frac{1+\varepsilon}{2\delta e s}\right)^{\frac{1+\varepsilon}{2}} $ and $ e^{-s\delta x^2} $
is bounded by $1$, we can put together
the above terms using Theorem~\ref{holder}(ii) and get that
\[
 \left\lVert (F|\Dd|^{1+\varepsilon}+R)e^{-s\Dd^2}\right\rVert _{s^{-1}} \leq
\left( \left(\frac{1+\varepsilon}{2\delta e s}\right)^{\frac{1+\varepsilon}{2}}
   \left\lVert F\right\rVert  + \left\lVert R\right\rVert \right)
\left( \tau (e^{-(1-\delta)\Dd^2} ) \right)^s \mbox{ .}
\]
Keeping in mind that at most $k$ of the $F_j$'s are non-zero,
 we get
\begin{eqnarray*}
 |\left\langle F_0 |\Dd|^{1+\varepsilon}  +R_0,\ldots,F_n |\Dd|^{1+\varepsilon}   +R_n \right\rangle ^n _\Dd| \hspace{4cm}
\\
\leq \tau (e^{-(1-\delta)\Dd^2} )\prod ^n _{j=0} 
(\left\lVert F_j\right\rVert+\left\lVert R_j\right\rVert ) 
   \left(\frac{1+\varepsilon}{2\delta e }\right)^{\frac{1+\varepsilon}{2}\cdot k}    
\int_{\Delta_n}
\left(s_0\ldots s_{k-1}\right)^{-\frac{1+\varepsilon}{2}} d^n s
\end{eqnarray*}


Along with the estimates
\begin{eqnarray*}
\left(\frac{1+\varepsilon}{2\delta e }\right)^{\frac{1+\varepsilon}{2}\cdot k}    
 \leq \left(\frac{1}{\delta e }\right)^{ k} 
\end{eqnarray*}  and  \begin{eqnarray*}
\int_{\Delta_n}
\left(s_0\ldots s_{k-1}\right)^{-\frac{1+\varepsilon}{2}} d^n s
\leq \left(\frac{2}{1-\varepsilon}\right)^k\frac{1}{(n-k)!} \mbox{ ,}
\end{eqnarray*} the proof is complete.
\end{proof}

The above norm estimate immediately implies that 
$\left\lVert \Ch^n(\Dd) \right\rVert < \frac{1}{n!} \tau(e^{-(1-\delta)\Dd^2}) C^n$.
Therefore,  $\Ch^\bullet(\Dd)$ is an entire cochain when  $\tau(e^{-(1-\delta)\Dd^2})<\infty$, which is exactly the
weakly $\theta$-summable condition.

\begin{proof}[Proof of Theorem~\ref{jloisacocycle}]
Lemma~\ref{getzler}
 guarantees that $ \Ch^\bullet(\Dd)$ is entire.
What remains to check is that  $ \Ch^\bullet(\Dd)$ is
$(b+B)$ closed. We adopted the computation in \cite{jlo} to
the Type II case.

We compute $\Ch^n(\Dd)$ paired with $b(a_0,\ldots,a_{n+1})_{n+1}$.
\begin{eqnarray*}
\left( \Ch^n(\Dd),b(a_0,\ldots,a_{n+1})_{n+1} \right)&=&
\left\langle 
a_0 a_1 ,[\Dd,a_2],\ldots, [\Dd,a_{n+1}]
\right\rangle ^n _{\Dd}
\\&&+ 
\sum_{j=1}^n (-1)^j 
\left\langle 
a_0, \ldots, [\Dd,a_j a_{j+1}],\ldots
\right\rangle ^n _{\Dd}
\\&&+ (-1)^{n+1}
\left\langle 
a_{n+1} a_0, [\Dd,a_1],\ldots, [\Dd,a_n]
\right\rangle ^n _{\Dd}
\\ &=&
\left\langle 
a_0 a_1 ,[\Dd,a_2],\ldots 
\right\rangle ^n _{\Dd}
- 
\left\langle 
a_0 ,a_1 [\Dd,a_2],\ldots 
\right\rangle ^n _{\Dd}
\\&&+ 
\sum_{j=2}^n (-1)^{j-1} \left(
\left\langle 
a_0, \ldots, [\Dd,a_{j-1} ]a_{j},\ldots
\right\rangle ^n _{\Dd} \right.
\\&& \hspace{2cm} \left.
+
\left\langle 
a_0, \ldots,a_{j} [\Dd,a_{j+1} ],\ldots
\right\rangle ^n _{\Dd}
\right)
\\&& +(-1)^{n}
\left(
\left\langle
a_0,[\Dd,a_1],\ldots,[\Dd,a_n]a_{n+1}
\right\rangle ^n _{\Dd} \right.
\\&& \hspace{2cm} \left.
-
\left\langle 
a_{n+1} a_0, [\Dd,a_1],\ldots, [\Dd,a_n]
\right\rangle ^n _{\Dd}\right)
\\ & \stackrel{\ref{misc}(4)}{=} &
\sum_{j=1}^{n+1} (-1)^{j-1} 
\left\langle
a_0, \ldots,[\Dd^2,a_j],\ldots 
\right\rangle ^{n+1} _{\Dd}\mbox{ .}
\end{eqnarray*}
The last term forms a telescope sum and reduces to
\[
\left\langle
a_0 \Dd, [\Dd,a_1],\ldots
\right\rangle ^{n+1}_{\Dd}
+ (-1)^n 
\left\langle
a_0 , [\Dd,a_1],\ldots, [\Dd,a_{n+1}]\Dd 
\right\rangle ^{n+1}_{\Dd}
= -
\left\langle
[\Dd,a_0],\ldots,[\Dd,a_{n+1}]
\right\rangle ^{n+1}_{\Dd} \mbox{ .}
\]
Now apply Lemma~\ref{misc}(1)(2), one checks that
\[
\left\langle
[\Dd,a_0],\ldots,[\Dd,a_{n+1}]
\right\rangle ^{n+1}_{\Dd}=\left( \Ch^{n+2}(\Dd),B(a_0,\ldots,a_{n+1})_{n+1}
\right)\mbox{ .}
\]
Therefore,
$
b\Ch^n(\Dd)=-B \Ch^{n+2}(\Dd)$  
and $(b+B)\Ch^\bullet(\Dd)=0$.
The proof is complete.
\end{proof}
As a result, the JLO character defines an entire cyclic
 cohomology class called the JLO class.



\subsection{Homotopy invariance of the JLO class}

In this section, we will show that the cohomology class
 given by the JLO character is homotopy invariant.
As a consequence,
the JLO character descends to a well-defined map 
from (semi-finite) $\KKEI$-homology to entire cyclic cohomology.
We follow closely to work by Getzler and Szenes \cite{getzlereven}.


\begin{definition}
Let $V$ be an operator affiliated with $\Nn$.
Define the \textbf{contraction}
 $\iota(V)$ by $V$ to be
\begin{eqnarray*}
\iota(V)\left\langle F_0, \ldots, F_n \right\rangle^n_\Dd
:=\sum _{k=0}^n 
\left. \left\langle F_0,\ldots,F_k,V,F_{k+1},\ldots,F_n\right\rangle
\right. ^{n+1}_\Dd \mbox{ .}
\end{eqnarray*}
\end{definition}
\begin{definition}
Let $V$ be an operator affiliated with $\Nn$ such that
it has the same degree as $\Dd$, 
i.e. $\lvert \Dd \rvert _\chi = \lvert V \rvert _\chi$.
Define $\hCh^\bullet(\Dd,V)$ to be
given by the equation
 \begin{eqnarray*}
\lefteqn{ \left( \hCh^n(\Dd,V),(a_0,\ldots,a_n)_n \right)}
\\
&:=&
\sum ^{n+1} _{j=1} (-1) ^{j}  
\left\langle a_0,
[\Dd,a_1],\ldots,[\Dd,a_{j-1}],V
,\ldots,[\Dd,a_n]\right\rangle ^{n+1}_\Dd \mbox{ .}
 \end{eqnarray*}
\end{definition}

\begin{theorem}
\label{123}
Let $\KCYCLE$ be a weakly $\theta$-summable unbounded
Breuer-Fredholm module.
  \begin{enumerate}
\item 
$\hCh^\bullet (\Dd,V)$ is an entire cochain if $V=F|\Dd|^{1+\varepsilon}+R$ 
where $0\leq \varepsilon<1$, $F$ and $R$ are operators in $\Nn$.
\item 
Let $V$ be an operator affiliated with $\Nn$ such that
it has the same degree as $\Dd$, 
i.e. $\lvert \Dd \rvert _\chi = \vert V \rvert _\chi$.
Then
\begin{eqnarray}
 b\hCh^{n-1}(\Dd,V)+B\hCh^{n+1}(\Dd,V)
=
-\iota(\Dd V + V\Dd)
 \Ch^n(\Dd)
+
\alpha ^n (\Dd,V)
    \mbox{ ,} 
\end{eqnarray}
where 
 $\alpha^n (\Dd,V)$ is defined to be
\[
 \left( \alpha^n (\Dd,V), (a_0,\ldots,a_n) \right)
:=
\sum_{j=1} ^n 
\left \langle
a_0,[\Dd,a_1],\ldots,[V,a_j],\ldots,[\Dd,a_n]\right\rangle ^n _\Dd \mbox{ ,}
\]
 \end{enumerate}
\end{theorem}

\begin{proof}\mbox{}
 \begin{enumerate} 
%
 \item 
From Lemma~\ref{getzler} we have that
 \[
                       \left\lVert
\hCh^n(\Dd,V)\right\rVert
\leq
\left(\frac{2}{(1-\varepsilon)\delta e}\right)
\frac{(n+1)}{n!}
\tau (e^{-(1-\varepsilon)\Dd^2}) C^n
\]
Therefore, 
\[
 \sum_{n=0}^\infty \Gamma(\frac{n}{2})\left\lVert \hCh^{n}(\Dd,V)\right\rVert z^{n}
\]
defines an entire function in $z$ and 
$\hCh^\bullet(\Dd,V)$ is entire.
\item 
Recall that 
\begin{eqnarray*}
\lefteqn{
\left(\hCh^{n-1}(\Dd,V),(b_0,\ldots,b_{n-1})_{n-1} \right)
}\\&=&\sum_{j=1}^{n} (-1)^{j-1} 
\langle b_0, \ldots,[\Dd,b_{j-1}],V
,[\Dd,b_{j}],
\ldots, [\Dd,b_{n-1}] \rangle ^n _{\Dd} \mbox{ .}
\end{eqnarray*}
Denote by $E_j$ the cochain
 \begin{eqnarray*}
\left( E_j , (b_0,\ldots,b_{n-1})_{n-1} \right)
&:=& \langle b_0, \ldots,[\Dd,b_{j-1}],V
,[\Dd,b_{j}],
\ldots, [\Dd,b_{n-1}] \rangle ^n _{\Dd} \mbox{ ,}
\end{eqnarray*}
so that 
\[\hCh^{n-1}(\Dd,V)= \sum_{j=1}^n (-1)^{j} E_j \mbox{ .}\]

First we compute $E_j$ paired with $b(a_0, \ldots,a_n)_n$:
\begin{eqnarray*}
\left( E_j , b(a_0, \ldots,a_n)_n\right)
&=&
\left\langle a_0 a_1, \ldots,[\Dd,a_{j}],V
,[\Dd,a_{j+1}],
\ldots, [\Dd,a_{n}] 
\right\rangle ^{n} _{\Dd} \\
&&
+\sum_{k=1}^{j-1} (-1)^k
\left\langle a_0 , \ldots, [\Dd,a_k a_{k+1}]  ,\ldots,[\Dd,a_j],V,
\ldots, [\Dd,a_{n}] 
\right\rangle ^{n} _{\Dd} \\
&& +\sum_{k=j}^{n-1} (-1)^k
\left\langle a_0 , \ldots,
V,[\Dd,a_{j}],
\ldots, [\Dd,a_k a_{k+1}]  ,\ldots, [\Dd,a_{n}] 
\right\rangle ^{n} _{\Dd} \\
&& + (-1)^n
\left\langle a_n a_0 , \ldots,
V,[\Dd,a_{j}],
\ldots, [\Dd,a_{n-1}] 
\right\rangle ^{n} _{\Dd} \mbox{ .}
\end{eqnarray*}
By expanding the $[\Dd,a_k a_{k+1}]$ terms using the Leibniz rule
and re-ordering the sum, we get
\begin{eqnarray*}
\lefteqn{
\left( E_j , b(a_0, \ldots,a_n)_n\right)}\\
&=&
\left\langle a_0 a_1, \ldots,[\Dd,a_{j}],V,
\ldots
\right\rangle ^{n} _{\Dd}-
\left\langle a_0 ,a_1[\Dd,a_2], \ldots,[\Dd,a_{j}],V,
\ldots
\right\rangle ^{n} _{\Dd}\\
&&+
\sum_{k=2}^{j} (-1)^{k-1}
\left(
\left\langle
\ldots,[\Dd,a_{k-1}]a_k,\ldots,V ,\ldots
\right\rangle ^n _{\Dd}
-\left\langle
\ldots,a_k[\Dd,a_{k+1}],\ldots,V ,\ldots
\right\rangle ^n _{\Dd}
\right) \\
&&
+ (-1)^{j-1}
\left\langle
\ldots, a_j V,\ldots
\right\rangle ^{n}_{\Dd} 
 - (-1)^{j-1}
\left\langle
\ldots,  V a_j,\ldots
\right\rangle ^{n}_{\Dd} \\
&& +
\sum_{k=j}^{n-1}(-1)^{k-1}
\left(
\left\langle 
\ldots,V, \ldots, [\Dd,a_{k-1}]a_k,\ldots 
\right\rangle ^n _{\Dd}
-
\left\langle 
\ldots,V,\ldots, a_k [\Dd,a_{k+1}],\ldots
\right\rangle ^n _{\Dd} \right)\\
&&+
(-1)^{n-1} \left(
\left\langle 
\ldots,V, \ldots, [\Dd,a_{n-1}]a_n
\right\rangle ^n _{\Dd}
- 
\left\langle 
a_n a_0,\ldots,V, \ldots, [\Dd,a_{n-1}]
\right\rangle ^n _{\Dd}\right)\mbox{ .}
\end{eqnarray*}
We are now in the setting to apply Lemma~\ref{misc}(4) to obtain
\begin{eqnarray*}
\left(  E_j , b(a_0, \ldots,a_n)_n\right) 
&=&
(-1)^j \left\langle \ldots,[V,a_j],\ldots  \right\rangle^n_{\Dd}\\
&& + \sum_{k=1}^{j} (-1)^{k-1}
\left\langle 
a_0,\ldots,[\Dd^2,a_k],\ldots,V,\ldots
\right\rangle ^{n+1}_{\Dd}
\\ &&+
\sum_{k=j}^{n} (-1)^{k-1}
\left\langle 
a_0,\ldots,V,\ldots,[\Dd^2,a_k],\ldots
\right\rangle ^{n+1}_{\Dd} \mbox{ .}
\end{eqnarray*}
From the facts that $[\Dd^2,a_k]=\Dd[\Dd,a_k]+[\Dd,a_k]\Dd$ and
$\Dd$ commutes with $e^{-s_k \Dd^2}$, one observes the above
forms a telescope sum and reduces to the following.
\begin{eqnarray*}
\lefteqn{
\left( E_j , b(a_0, \ldots,a_n)_n\right)}\\
&=&
(-1)^j \left\langle \ldots,[V,a_j],\ldots  \right\rangle^n_{\Dd}
+
\left \langle  a_0\Dd,\ldots,[\Dd,a_j],V,\ldots \right\rangle ^{n+1}_{\Dd}
\\&& + (-1)^{j-1}
 \left \langle \ldots,[\Dd,a_j] , \Dd V ,\ldots \right\rangle ^{n+1}_{\Dd}
+ (-1)^{j-1}
\left \langle \ldots,   V\Dd, [\Dd,a_j] ,\ldots \right\rangle ^{n+1}_{\Dd}
\\&&+
 (-1)^{n-1}\left\langle \ldots,V,[\Dd,a_j],\ldots,[\Dd,a_n]\Dd\right \rangle ^{n+1}_{\Dd}\\
&=&
(-1)^j \left\langle \ldots,[V,a_j],\ldots  \right\rangle^n_{\Dd}
+
\left \langle  a_0\Dd,\ldots,[\Dd,a_j],V,\ldots \right\rangle ^{n+1}_{\Dd}
\\&&+
 \left\langle \Dd a_0, \ldots,V,[\Dd,a_j],\ldots,[\Dd,a_n]\right \rangle ^{n+1}_{\Dd}
\\&& + (-1)^{j-1}
 \left \langle \ldots,[\Dd,a_j] , \Dd V ,\ldots \right\rangle ^{n+1}_{\Dd}
+ (-1)^{j-1}
\left \langle \ldots,   V\Dd, [\Dd,a_j] ,\ldots \right\rangle ^{n+1}_{\Dd}
\mbox{ .}
\end{eqnarray*}
Now we sum over $j$ with the appropriate sign:
\begin{eqnarray}
\lefteqn{
\left(\sum_{j=1}^n (-1)^{j}E_j , b(a_0, \ldots,a_n)_n\right)}\\
&=& \label{uselesseqn1}
 \sum_{j=1}^n \left\langle \ldots,[V,a_j],\ldots  \right\rangle^n_{\Dd}
\\&& -\label{uselesseqn2}
\sum_{j=0}^{n} 
 \left \langle \ldots,[\Dd,a_j] , \Dd V +V\Dd,\ldots \right\rangle ^{n+1}_{\Dd}
\\&& -\label{uselesseqn3}
\sum_{j=0}^{n} (-1)^{j}
 \left\langle [\Dd ,a_0] ,\ldots,[\Dd,a_{j}],V,\ldots,[\Dd,a_n] \right \rangle ^{n+1}_{\Dd} \mbox{ .}
\end{eqnarray}
Equations ~\eqref{uselesseqn1} and ~\eqref{uselesseqn2}
give $\alpha^n(\Dd,V)$ and $-\iota(\Dd V+V\Dd)\Ch^n(\Dd)$
respectively.
By using Lemma~\ref{misc}(1)(2), each summand in
 Equation~\eqref{uselesseqn3} 
can be written as
\begin{eqnarray}
\lefteqn{ (-1)^{j}
 \left\langle [\Dd ,a_0] ,\ldots,[\Dd,a_j] ,V,\ldots
\right \rangle ^{n+1}_{\Dd}}\\
&\stackrel{\ref{misc}(2)}{=}&
(-1)^j
\sum_{k=0}^{j+1}
\left\langle 
\ldots,[\Dd,a_{k-1}],1,
 \ldots, V ,\ldots
\right\rangle ^{n+2}_{\Dd}
\\&&+ (-1)^j
\sum_{k=j+1}^{n}
\left\langle 
\ldots,
 V ,\ldots,1,[\Dd,a_k] , \ldots
\right\rangle ^{n+2}_{\Dd}
\\&\stackrel{\ref{misc}(1)}{=}&
\label{moreuselesseqn1}
\sum_{k=0}^{j+1}
\left(
(-1)^{j+2-k}
E_{j+2-k} , (-1)^{nk} (1,a_k,\ldots, a_{k-1})_{n+1}
\right)
\\&&  \label{moreuselesseqn2} +
\sum_{k=j+1}^{n}
\left(
(-1)^{n-k+j+3} E_{n-k+j+3}, (-1)^{nk} (1,a_k,\ldots, a_{k-1})_{n+1}
\right) \mbox{ .}
\end{eqnarray}
Now we sum over  $j=0$ to $n$, and do 
  a change of indices of $i=j-k+1$ 
for Equation~\eqref{moreuselesseqn1} and $i=n-k+j+1$
for Equation~\eqref{moreuselesseqn2}, 
then Equation~\eqref{uselesseqn3} becomes
\begin{eqnarray*}
\lefteqn{\sum_{i=0}^n \left( (-1)^{i+1}E_{i+1},\sum_{k=0}^{n-i}
(-1)^{nk}
(1,a_k,\ldots,a_{k-1})_{n+1}\right)} \\&&
+
\sum_{i=0}^n \left( (-1)^{i+2}E_{i+2},\sum_{k=n-i}^{n}
(-1)^{nk}(1,a_k,\ldots,a_{k-1})_{n+1}\right) \mbox{ ,}
\end{eqnarray*}
which equals $\left(\hCh^{n+1}(\Dd,V),B(a_0,\ldots,a_n)_n \right)$.
Hence we have obtained 
\[
b\hCh^{n-1}(\Dd,V)=-\iota(\Dd V + V\Dd) \Ch^n(\Dd)+\alpha^n(\Dd,V)
- B\hCh^{n+1}(\Dd,V),
\]
which is the desired result.
\end{enumerate}
\end{proof}

Suppose that $\Dd_t$ is a $t$-parameter family of operators so that it
defines a differentiable family of  weakly $\theta$-summable
unbounded Breuer-Fredholm modules $\left(\rho,\Nn,\Dd_t  \right)$.
Namely, $\Dd_t$  is a $t$-parameter family of self-adjoint operators on $\Hh$ 
with common domain of definition
so that the following is satisfied:
\begin{itemize}
\item
$\Dd_t$ is  affiliated with $\Nn$ for all $t$,
\item
For all
 $a\in A$, $[\Dd_t , \rho(a) ]$ is 
a  norm-differentiable family of operators in  $  \Nn$  ,
 and 
there is a constant $C$ for each compact interval such that
$ \left\lVert [\Dd_t,\rho(a)]\right\rVert \leq C \left\lVert a\right\rVert$,
\item
$(1+\Dd_t^2)^{-1/2}$ is a  norm-differentiable
 family of operators in $\Kk$,
\item
There exists a $u\in(0,1)$ such that
$\tau(e^{-u\Dd_t^2})$ is  bounded for each compact interval.
 \end{itemize}

If $(\rho,\Nn,\Dd_t)$ is equipped with a $\mathbb{Z}_2$ grading 
$\chi\in\Nn$ so that $\rho(a)$ is even for all $a\in A$ and
$\Dd_t$ is odd for all $t$, then similarly we call the
family of Breuer-Fredholm modules $(\rho,\Nn,\Dd_t)$ \emph{even}.

The differentiable families of unbounded operators in our discussion will
 often be  ``functions'' of $\Dd$, hence we do not alter the spectral projections.
For more general notions of differentiable family of unbounded operators, readers may refer to
\cite{diffdirac}.


\begin{lemma}[\cite{oddjloT2}]
Let $(\rho,\Nn,\Dd_t)$
be a differentiable family of weakly $\theta$-summable
unbounded Breuer-Fredholm modules, and
 $F_0,\ldots,F_n$ be operators affiliated with $\Nn$, then
 \label{duhammel}
\[
\frac{d}{dt}\left\langle F_0,\ldots,F_n \right \rangle ^n _{\Dd_t}
= -
\sum _j ^n \left\langle F_0,\ldots,F_j,
 \Dd_t \dot{\Dd_t} + \dot{\Dd_t}\Dd_t
,F_{j+1},\ldots,F_n 
\right\rangle ^{n+1} _{\Dd_t}
\mbox{ .}\]
\end{lemma}
\begin{theorem}
\label{cobound}
 If $\dot{\Dd_t}=F_t|\Dd_t|^{1+\varepsilon}+R_t$
 for $0\leq \varepsilon <1$ and 
 $F_t, R_t \in \Nn$ are continuous families of operators 
that are uniformly bounded in $t$ then
$\hCh^\bullet(\Dd_t,\dot{\Dd_t})$ is an entire cochain and for every $n$
\[
  \frac{d \Ch^n(\Dd_t)}{dt} = b \hCh^{n-1}(\Dd_t,\dot{\Dd_t})+B\hCh^{n+1}(\Dd_t,\dot{\Dd_t})\mbox{ .}
\]

\end{theorem}

\begin{proof}
 By applying the Leibniz rule on $\frac{d}{dt}\Ch^n(\Dd_t)$,
we will get terms containing  
$\frac{d}{dt}e^{-(t_{j+1}-t_{j})\Dd_t^2}$
and terms containing $\frac{d}{dt}[\Dd_t,a_j]$.
By Lemma~\ref{duhammel}, the former collects into 
$\iota(\Dd_t \dot{\Dd_t} + \dot{\Dd_t} \Dd_t)\Ch^n(\Dd_t)$,
while the latter collects into $\alpha ^n (\Dd_t,\dot{\Dd_t})$.
Hence together with Theorem~\ref{123}(3)
\begin{eqnarray*}
 \frac{d\Ch^n(\Dd_t)}{dt}
&=&\iota(\Dd \dot{\Dd_t} + \dot{\Dd_t} \Dd)
\Ch^n (\Dd_t )
+
\alpha ^n (\Dd_t,\dot{\Dd_t})
\\&=& b\hCh^{n-1}(\Dd_t,\dot{\Dd_t})+B\hCh^{n+1}(\Dd_t,\dot{\Dd_t})  \mbox{ .}
\end{eqnarray*}
The fact that $\hCh^\bullet(\Dd_t,\dot{\Dd_t})$
is entire follows from Lemma~\ref{getzler} and the uniform
boundedness  of $F_t$ and $R_t$.
The result is obtained.
\end{proof}

The following Proposition gives a stability of bounded
perturbation of weakly $\theta$-summable unbounded Breuer-Fredholm
modules.
It is Theorem C in \cite{getzlereven}.
\begin{proposition}
\label{thmC}
 For a 
weakly  $\theta$-summable unbounded Breuer-Fredholm module
 $\KCYCLE$, and an operator $V\in\Nn$ 
such that $V$ has the same degree as $\Dd$, i.e. 
$\lvert V \rvert _\chi = \lvert \Dd \rvert _\chi$.
 Then  $\left(\rho, \Nn,\Dd+V \right)$ is
again a weakly $\theta$-summable unbounded Breuer-Fredholm
 module and
\[
 \tau\left( e^{-(1-\varepsilon/2)(\Dd+V)^2} \right) \leq 
e^{(1+2/\varepsilon)\left\lVert V\right\rVert^2} \cdot \tau \left( e^{-(1-\varepsilon)\Dd^2} \right) \mbox{ .}
\]
\end{proposition}
\begin{proof}
 It is obvious that 
\[
 \left\lVert [\Dd+V,a]\right\rVert\leq (C+2\left\lVert V\right\rVert)\left\lVert a\right\rVert \mbox{ ,}
\]
hence if we obtain $\tau\left( e^{-(1-\varepsilon/2)(\Dd+V)^2} \right) \leq 
e^{(1+2/\varepsilon)\left\lVert V\right\rVert^2} \cdot \tau \left( e^{-(1-\varepsilon)\Dd^2} \right)$, we are done.

Observe that if $A$ and $B$ are positive operators, then
\[
 \tau( e^{-A-B} )\leq \tau(e^{-A})\mbox{ .}
\]
We proceed by introducing the operators
\begin{eqnarray*}
 A&=&(1-\varepsilon)\Dd^2 \mbox{ , }\\
B&=&\frac{\varepsilon}{2}\Dd^2 + (1-\frac{\varepsilon}{2})\left(\Dd V + V\Dd + V^2\right) +
(1+\frac{2}{\varepsilon})\left\lVert V\right\rVert^2 \mbox{ .}
\end{eqnarray*}
$A$ is a positive operator, and to see that $B$ is also positive, we use the fact that
\[
 -(\Dd V+V\Dd)\leq \frac{\varepsilon}{2}\Dd^2 + \frac{2}{\varepsilon} V^2 \leq
 \frac{\varepsilon}{2}\Dd^2 + \frac{2}{\varepsilon} \left\lVert V\right\rVert^2 \mbox{ .}
\]
Therefore, 
\begin{eqnarray*}
 \tau\left( e^{-(1-\frac{\varepsilon}{2})\Dd^2 - (1-\frac{\varepsilon}{2})\left(\Dd V + V\Dd + V^2\right) -
(1+\frac{2}{\varepsilon})\left\lVert V\right\rVert^2} \right)=\tau\left( e^{-A-B} \right)
&\leq &\tau\left( e^{-A}\right)=\tau\left( e^{(1-\varepsilon)\Dd^2}\right)\\
   \tau\left( e^{-(1-\varepsilon/2)(\Dd+V)^2} \right)& \leq &
e^{(1+2/\varepsilon)\left\lVert V\right\rVert^2} \cdot \tau \left( e^{-(1-\varepsilon)\Dd^2} \right)\mbox{ ,}
\end{eqnarray*}
and the result is obtained.
\end{proof}

\subsection{Index pairing in (co)homology}

This section will show that the JLO character for a
weakly
 $\theta$-summable even unbounded Breuer-Fredholm module produces an index formula. For the odd case, we refer to a paper by Carey and Phillips \cite{oddjloT2}, who developed the JLO character in the Type II setting.

\begin{theorem}[\cite{oddjloT2}] 
Let $\KCYCLE$ be an odd weakly $\theta$-summable unbounded Breuer-Fredholm module over $ A$ and
$u\in M_N( A)$ be a unitary, then
\[
\langle  [\KCYCLE], [u] \rangle =\left( [\Ch^- (\Dd)] , [\CCHERN_-(u)] \right) \mbox{ ,}
\] 
where the angle bracket on the left is the spectral flow pairing \cite{oddjloT2}
 and the round bracket on the right is the (co)homology
pairing.
\end{theorem}

\begin{theorem}
Let $\KCYCLE$ be an even weakly $\theta$-summable unbounded Breuer-Fredholm module over $ A$ and
$p\in M_N( A)$ be a projection, then
\[
\langle  [\KCYCLE], [p] \rangle =\left( [\Ch^+ (\Dd)] , [\CCHERN_+(p)] \right) \mbox{ ,}
\] 
where the angle bracket on the left is the index pairing and the round bracket on the right is the (co)homology
pairing.
\end{theorem}

\begin{proof}

It suffices to prove that
\begin{eqnarray*}
\IND (p^- (\Dd\otimes 1_N) p^+ ) &=& \left( \Ch^+ (\Dd) , \CCHERN_+(p) \right) \mbox{ .}
\end{eqnarray*}
It follows from the definition of (co)homology that the above equality will descend to the result
stated in the theorem.


For any projection $p \in  A$, 
 one can deform $\Dd$ to $\left(p\Dd p + (1-p)\Dd(1-p)\right)$ via
the homotopy $\Dd_t=\Dd+t(2p-1)[\Dd,p]$ where $t\in[0,1]$.
As $\dot{\Dd_t}=(2p-1)[\Dd,p]$ is odd and in $\Nn$, by Proposition~\ref{thmC}, 
$\left(\rho, \Nn,\Dd_t \right)$ is a differentiable family of 
weakly $\theta$-summable unbounded Breuer-Fredholm modules.
 By Theorem~\ref{cobound},
$\Ch^+(\Dd)$ and $\Ch^+\left(p\Dd p + (1-p)\Dd(1-p)\right)$ are cohomologous. Specifically, 
\begin{eqnarray*}
 \lefteqn{ 
\Ch^+\left(p\Dd p + (1-p)\Dd(1-p) \right)
- \Ch^+(\Dd) } \\&=&
\Ch^+(\Dd_1)- \Ch^+(\Dd_0)
\\
&=&(b+B)\int_0^1 \hCh^+(\Dd_t,\dot{\Dd_t}) \mbox{ .}
\end{eqnarray*}
Therefore, 
\begin{eqnarray*}
 \left( \Ch^+ (\Dd) , \CCHERN_+(p) \right) &=&
\left( \Ch^+\left(p\Dd p + (1-p)\Dd(1-p) \right)  ,  \CCHERN_+(p) \right)  
\\&& \hspace{2cm} - \left( (b+B) \int_0^1 \hCh^+(\Dd_t,\dot{\Dd_t}) , \CCHERN_+(p) \right)   \\
&=& 
\left( \Ch^0 \left(p\Dd p + (1-p)\Dd(1-p) \right)  ,  \CCHERN_+(p) \right)  
\end{eqnarray*}
where the last equality follows from the fact that
$[\Dd_1,p]=[p\Dd p + (1-p)\Dd(1-p) , p ]=0$ and $ \CCHERN_+(p)$ is closed.
Hence the pairing $\left( \Ch^+ (\Dd) , \CCHERN_+(p) \right)$ yields
the McKean-Singer index formula
\[
 \tau (\chi p e^{-\Dd^2})\mbox{ .}
\]
The fact that the McKean-Singer index formula produces the desirable index 
$\IND (p^- \Dd p^+ )$ is proved in
\cite{cprs3}.
If $p$ is a projection in $M_N( A)$, one extends $\Dd$ to $\Dd\otimes 1_N$ and $\tau$ to $\tau\otimes \operatorname{Tr}$, where
$\operatorname{Tr}$ is the matrix trace from $M_N(\mathbb{C})\rightarrow \mathbb{C}$,
the result follows.
\end{proof}

\section{Reduction from JLO character to Connes-Chern character}

For every $p$-summable unbounded Breuer-Fredholm module
there is a canonically associated 
$p$-summable bounded
Breuer-Fredholm module.
 Using techniques from \cite{chernreduct}, this section connects the previous two sections  by 
showing that the JLO character of a  $p$-summable unbounded Breuer-Fredholm module 
and the Connes-Chern character of its associated $p$-summable Breuer-Fredholm module 
define the same class
in entire cyclic cohomology.

Most of the work in this section is adapted from \cite{chernreduct} and \cite{greenbook}.


\subsection{JLO character for $p$-summable unbounded Breuer-Fredholm modules}

As shown in Lemma~\ref{ptheta},  $p$-summable unbounded Breuer-Fredholm modules
are also $\theta$-summable (in particular weakly $\theta$-summable).
 Therefore, the JLO character of a $p$-summable unbounded Breuer-Fredholm 
is defined.

Given  a $p$-summable unbounded Breuer-Fredholm module $\KCYCLE$, its 
JLO class defined by the
JLO character $\Ch^\bullet(\Dd)$ has a representative which consists of only finitely many terms.

For convenience, we denote 
\begin{eqnarray}
\label{eqn:retractedJLO}
 \cch^n_t(\Dd):= \Ch^{\leq n}(t\Dd) + B\int_0^t\hCh^{n+1}(u\Dd,\Dd)du \mbox{ .}
\end{eqnarray}
Here $\Ch^{\leq n}(t\Dd)$ means that we discard the terms greater than $n$ in $\Ch^{\bullet}(t\Dd)$. That is,
\[
\Ch^{\leq n}(t\Dd) := 
\left\{
 \begin{array}{ll}
 \sum_{k=1}^{2k \leq n} \Ch^{2k}(t\Dd) 
&  \mbox{ when  $(\rho,\Nn,\Dd)$ is even}\\
\sum_{k=1}^{2k+1 \leq n} \Ch^{2k+1}(t\Dd)
&  \mbox{ when $(\rho,\Nn,\Dd)$ is odd}
\end{array} 
\right.
\mbox{ .}
\]

\begin{proposition}
\label{finitejlo}
Given a $p$-summable unbounded Breuer-Fredholm module $\KCYCLE$,
its JLO-cocycle
$ \Ch^\bullet(\Dd) $ is cohomologous to
$\cch^n_t(\Dd)$
 for
$t\in[1,\infty), n>p$.
When $\Dd$ is invertible, 
 $B\int_0^\infty \hCh^{n+1}(u\Dd,\Dd)du$ is a well-defined entire cochain and
is cohomologous to $ \Ch^\bullet(\Dd) $.
\end{proposition}
\begin{proof} The proof will make use of Theorem~\ref{123}(2) twice.
 \begin{eqnarray*}
 \Ch^\bullet(\Dd)
&=& \Ch^\bullet(t\Dd)+ (b+B)\int_1^t\hCh ^\bullet(u\Dd,\Dd)du\\
&=& \Ch^{\leq n}(t\Dd) + \Ch^{\geq n+2}(s\Dd) + \left(
 \Ch^{\geq n+2}(t\Dd) - \Ch^{\geq n+2}(s\Dd) \right) 
\\&& \hspace{2cm}+ 
(b+B)\int_1^t\hCh^\bullet(u\Dd,\Dd)du \\
&=& \Ch^{\leq n}(t\Dd) + \Ch^{\geq n+2}(s\Dd)
\\&& \hspace{2cm} + \left(
 - b\int_s^t \hCh^{\geq n+1}(u\Dd,\Dd)du - B\int_s^t\hCh^{\geq n+3}(u\Dd,\Dd)du \right)\\
&&\hspace{2cm}+ (b+B)\int_1^t\hCh^\bullet(u\Dd,\Dd)du \\
&=& \Ch^{\leq n}(t\Dd) + B\int_s^t\hCh^{n+1}(u\Dd,\Dd)du + \Ch^{\geq n+2}(s\Dd) \\
&&\hspace{2cm}+  (b+B)\left(\int_1^t\hCh^\bullet(u\Dd,\Dd)du
 -  \int_s^t \hCh^{\geq n+1}(u\Dd,\Dd)du  \right) \mbox{ .}
\end{eqnarray*}
Since $\tau ( e^{-(1-\varepsilon)s^2\Dd^2}) = O(s^{-p})$ by Lemma~\ref{ptheta}, it follows from Lemma~\ref{getzler} that for $m>p$
\begin{eqnarray*}
\label{inftylimit}
  \left\lVert \Ch ^m (s\Dd)\right\rVert &=& O(s^{m-p}) \\
\left\lVert \hCh^m(u\Dd,\Dd)\right\rVert &=&O(u^{m-p})\mbox{ ,}
\end{eqnarray*}
so that
$\lim_{s\searrow 0} \Ch ^m (s\Dd) = 0$ in norm and 
$\hCh^m(u\Dd,\Dd)$ is integrable from $0$ to $t$ for $m>p$.
Hence 
\begin{eqnarray*}
 \Ch^\bullet(\Dd) &=& \Ch^{\leq n}(t\Dd) + B\int_0^t\hCh^{n+1}(u\Dd,\Dd)du 
\\ && \hspace{2cm} +
(b+B)\left(\int_1^t\hCh^\bullet(u\Dd,\Dd)du
 -  \int_0^t \hCh^{\geq n+1}(u\Dd,\Dd)du  \right)\mbox{ .}
\end{eqnarray*}
Now suppose $\Dd$ is invertible. Then
 $\lambda:=\inf\left\{\sigma(\Dd^2)\right\} >0$
 where $\sigma(\Dd^2)$ is the spectrum of $\Dd^2$.
From Lemma~\ref{getzler} and (the proof of) Lemma~\ref{ptheta}, we get that
\begin{eqnarray*}
 \| \Ch^{r} (t\Dd) \| &\leq&
\frac{t^r}{r!}\cdot \tau\left(e^{-(1-\delta)t^2\Dd^2 }\right) C^r \\
&\leq&   \frac{t^r}{r!}
\left\lVert e^{-(1-\delta)t^2\Dd^2 /2 } \right\rVert
 \tau\left(e^{-(1-\delta)t^2\Dd^2 /2} \right) C^r  \\
& \leq&  \left(
\frac{t^r}{r!} e^{-(1-\delta)t^2\lambda^2 /2}\right)
 \tau\left(e^{-(1-\delta)t^2\Dd^2 /2} \right) C^r    \mbox{ ,}
\end{eqnarray*}
and
\begin{eqnarray*}
 \| \hCh^{n+1}(u\Dd,\Dd) \| &\leq&\frac{2u^{n+1}}{n!\sqrt{e \delta}}
\cdot \tau\left(e^{-(1-\delta)u^2\Dd^2 }\right) C^{n+1} \\
&\leq& 
\frac{2u^{n+1}e^{-(1-\delta)u^2\lambda ^2/2 }}{n!\sqrt{e \delta}}
\cdot \tau\left(e^{-(1-\delta)u^2\Dd^2/2 }\right) C^{n+1}
\\ &\leq&
\frac{2u^{n+1-p}e^{-(1-\delta)u^2\lambda ^2/2 }}{n!\sqrt{e \delta}}
\left(\frac{p}{e(1-\delta)}\right)^{p/2}
\left\lVert \lvert \Dd \rvert^{-1} \right \rVert _{p}^p C^{n+1}
 \mbox{ .}
\end{eqnarray*}
The norm of $ \tau\left(e^{-(1-\delta)t^2\Dd^2 /2} \right)$ is uniformly bounded
for $t\in [1,\infty)$,
and the term \newline
$u^{n+1-p}e^{-(1-\delta)u^2\lambda ^2/2 }$ is
integrable from $0$ to $\infty$,
 therefore
 the limit for $t\rightarrow\infty$ exists in norm
for  $\Ch^{r}(t\Dd)$ and $\int_0^t\hCh^{n+1}(u\Dd,\Dd)du$.
 In particular, 
$\lim_{t\rightarrow\infty} \Ch^{\leq n}(t\Dd)=0 $ in norm.
Thus,
\begin{eqnarray*}
  \Ch^\bullet(\Dd) & =&   B\int_0^\infty \hCh^{n+1}(u\Dd,\Dd)du 
\\ &&\hspace{2cm} +
(b+B)\left(\int_1^\infty \hCh^\bullet(u\Dd,\Dd)du
 -  \int_0^\infty \hCh^{\geq n+1}(u\Dd,\Dd)du  \right)\mbox{ .}
\end{eqnarray*}
\end{proof}

From the proof of Proposition~\ref{finitejlo},
\begin{equation}
 \lim_{t\rightarrow\infty} \cch^n_t(\Dd) =  B\int_0^\infty\hCh^{n+1}(u\Dd,\Dd)du \mbox{ .}
\end{equation}

\subsection{From unbounded to bounded Breuer-Fredholm modules}
\label{unbounded2bounded}
For every $p$-summable unbounded Breuer-Fredholm module
there is a canonically associated 
$p$-summable bounded
Breuer-Fredholm module. We will go through a concrete 
construction of such a bounded Breuer-Fredholm module
from an unbounded one when $\Dd$ is invertible, and remove
the invertibility assumption at the end of the section.
Most of the work in this section is adopted from \cite{chernreduct} and \cite{fromboundedtounbounded}.

Given an unbounded Breuer-Fredholm module $\KCYCLE$ with $\Dd$ invertible,
 there is an associated bounded Breuer-Fredholm module
$\FREDMOD$ by taking $F=\Dd |\Dd|^{-1}$.
We will follow a technique in \cite{fromboundedtounbounded,greenbook} 
to show
that if $\KCYCLE$ is $p$-summable, then so is $\FREDMOD$.

\begin{proposition}
\label{pp}
 If $\KCYCLE$ is a $p$-summable unbounded Breuer-Fredholm module with $\Dd$ invertible,
then $ \lvert \Dd \rvert ^{-(1-\alpha)} 
 \in \Ll^{\frac{p}{1-\alpha}}$ and 
 $ [\Dd |\Dd|^{-\alpha} ,a]\in \Ll^{\frac{p}{\alpha}}$ for $\alpha\in [0,1]$
with $\Ll^{\infty}:=\Nn$.
Furthermore, there exists a constant $C_p$ depending on $p$ and $\Dd$ so that
\[
\left \lVert  [\Dd |\Dd|^{-\alpha} ,a] \right \rVert
 _{\frac{p}{\alpha}}
\leq \left \lVert  [\Dd   ,a] \right \rVert C_p^{\alpha} \mbox{ .}
\]
\end{proposition}
\begin{proof}
It is immediate that 
 $ \lvert \Dd \rvert ^{-(1-\alpha)}  
\in 
\Ll^{\frac{p}{1-\alpha}}$.
Using the spectral formula for the
power $\alpha\in (0,1]$
 of a positive operator
\[
 H^{-\frac{\alpha}{2}}=\frac{1}{C_\alpha}\int_0^\infty 
 (\lambda +H)^{-1} 
  \lambda^{-\frac{\alpha}{2}} d\lambda 
\]
with $C_\alpha=\int_0^\infty (1+ x)^{-1}x^{-\frac{\alpha}{2}} dx$,
 we compute for $H=\Dd^2$ and $b$ in a Banach $*$-algebra $A$,
\begin{eqnarray*}
 [\Dd \lvert \Dd \rvert ^{-\alpha},b]&=& 
[\Dd,b](\Dd^2)^{-\alpha/2}+ \Dd[ (\Dd^2)^{-\alpha/2},b] \\
&=&\frac{1}{C_\alpha} \int_0^\infty\left( [\Dd,b](\lambda+ \Dd^2)^{-1}
+ \Dd[ (\lambda +\Dd^2)^{-1},b]
\right) \lambda^{-\alpha/2} d\lambda \\
&=& \frac{1}{C_\alpha} \int_0^\infty 
\left( [\Dd,b](\lambda+ \Dd^2)^{-1}- \Dd(\lambda +\Dd^2)^{-1}
[  (\lambda +\Dd^2),b](\lambda +\Dd^2)^{-1}
\right) \lambda^{-\alpha/2} d\lambda \\
&=& \frac{1}{C_\alpha} \int_0^\infty\left( (\lambda+\Dd^2)
(\lambda+  \Dd^2)^{-1}[\Dd,b](\lambda+ \Dd^2)^{-1} \right. \\
 &&\hspace{1cm}  -\left. \Dd(\lambda  +\Dd^2)^{-1}
\left(\Dd [\Dd ,b] + [\Dd ,b] \Dd\right)
(\lambda  +\Dd^2)^{-1}
\right) \lambda^{-\alpha/2} d\lambda  \\
&=& \frac{1}{C_\alpha} \int_0^\infty
\left( \lambda  (\lambda +\Dd^2)^{-1} 
[\Dd,b](\lambda +\Dd^2)^{-1}\right. \\ && \hspace{1cm} 
\left.- 
\Dd(\lambda  +\Dd^2)^{-1}[ \Dd,b](\lambda  +\Dd^2)^{-1}\Dd
\right) \lambda^{-\alpha/2} d\lambda \mbox{ .}
\end{eqnarray*}
If $[D,b]$ is self-adjoint, then the estimate 
$  -\lVert [\Dd,b] \rVert  \leq   [\Dd,b]  \leq \lVert [\Dd,b] \rVert$ 
and the fact that $\Dd$ is self-adjoint yield
\begin{eqnarray*}
- \lambda  (\lambda +\Dd^2)^{-1}\left\lVert [\Dd,b] \right\rVert  (\lambda  +\Dd^2)^{-1}  &\leq &
\lambda (\lambda+\Dd^2)^{-1}[\Dd,b]  (\lambda +\Dd^2)^{-1} \\
&\leq & \lambda  (\lambda  +\Dd^2)^{-1}\lVert[\Dd,b]\rVert(\lambda +\Dd^2)^{-1} 
\mbox{ ,}
\end{eqnarray*}
and
\begin{eqnarray*}
  -|\Dd|(\lambda  +\Dd^2)^{-1}\lVert [\Dd,b] \rVert  (\lambda  +\Dd^2)^{-1}|\Dd| &\leq & 
-\Dd(\lambda  +\Dd^2)^{-1}[\Dd,b]  (\lambda  +\Dd^2)^{-1}\Dd \\
&\leq&  |\Dd|(\lambda +\Dd^2)^{-1}\left\lVert[\Dd,b]\right\rVert(\lambda  +\Dd^2)^{-1}|\Dd| 
\mbox{ .}
\end{eqnarray*}
Therefore,
\begin{eqnarray*}
   -\left\lVert [\Dd,b] \right\rVert  | \Dd|^{-\alpha}   
\leq [\Dd \lvert\Dd \rvert ^{-\alpha},b] \leq  
\left\lVert [\Dd,b] \right\rVert 
 | \Dd|^{-\alpha}
\mbox{ ,}
\end{eqnarray*} 
and 
\begin{eqnarray}
\label{eqn:abc123temporary}
 \lVert [\Dd \lvert \Dd \rvert ^{-\alpha},b] \rVert _{\frac{p}{\alpha}}
 \leq \left\lVert [ \Dd,b] \right\rVert
\left( \left\lVert  | \Dd|^{-1} \right\rVert _p \right)^\alpha
\leq 
\left\lVert [ \Dd,b] \right\rVert \left(
\left\lVert  ( 1+ \Dd^2)^{-1/2} \right\rVert _p \right)^\alpha < \infty \mbox{ .}
\end{eqnarray}
Now for any $a\in A$,
 $[\Dd \lvert \Dd \rvert ^{-\alpha},a]=[\Dd \lvert \Dd \rvert ^{-\alpha},\frac{a-a^*}{2}]+i[\Dd \lvert \Dd \rvert ^{-\alpha},\frac{a+a^*}{2i}]$ with
 $[\Dd \lvert \Dd \rvert ^{-\alpha},\frac{a-a^*}{2}]$ and 
$[\Dd \lvert \Dd \rvert ^{-\alpha},\frac{a+a^*}{2i}]$ self-adjoint.
By Equation~\ref{eqn:abc123temporary},
$[\Dd \lvert \Dd \rvert ^{-\alpha},\frac{a-a^*}{2}],[\Dd \lvert \Dd \rvert ^{-\alpha},\frac{a+a^*}{2i}]\in\Ll^p$.
Thus  for $\alpha\in[0,1]$, 
\[
\left \lVert [\Dd \lvert \Dd \rvert ^{-\alpha},a]
\right \rVert \leq
\left \lVert  [\Dd ,a] \right \rVert 
 \left(
\left\lVert  ( 1+ \Dd^2)^{-1/2} \right\rVert _p \right)^\alpha \mbox{ ,}\]
and 
$[\Dd \lvert \Dd \rvert ^{-\alpha},a]\in\Ll^{p/\alpha}$.
 The proof is complete.
\end{proof}

\begin{corollary}
\label{pbeta}
If   $\KCYCLE$ is a $p$-summable unbounded Breuer-Fredholm module
with $\Dd$ invertible,
 then its associated Breuer-Fredholm module 
$\FREDMOD$ is also $p$-summable.
\end{corollary}

\subsection{From JLO character to Connes-Chern character}

 Let $\Dd_\alpha:=\Dd |\Dd|^{-\alpha}$ for $\alpha\in[0,1]$, then we have a homotopy between
the unbounded Breuer-Fredholm module  $\KCYCLE$ when $\alpha=0$ and 
its associated bounded Breuer-Fredholm module $\FREDMOD$
 when $\alpha=1$. 
%
%
%
%
When $\KCYCLE$ is $p$-summable and $n-1>p$, we will see that 
$\cch^n_t(\Dd_\alpha)$ (see Equation~\eqref{eqn:retractedJLO})
 defines an  
$\alpha$-family of entire cyclic cocycles for  $t>0$. 
Moreover, these  cocycles
 in fact live in the same entire cyclic cohomology class.
%

Note that our calculations include both the even and odd cases.

\begin{theorem}
\label{cchisexact}
The cochain
$\frac{d\cch^n_t(\Dd_\alpha)}{d\alpha}$ is exact for $\alpha\in[0,1]$ and 
$t\in[1,\infty]$.
Explicitly, it is the $(b+B)$-coboundary of
the entire cochain
\[
\int^t_0 
b\iota(\Dd_\alpha)\hCh^{n}(u\Dd_\alpha,u\dot{\Dd_\alpha})du -\hCh^{\leq n-1}(t\Dd_\alpha,t\dot{\Dd_\alpha})
 \mbox{ .}
\]

\end{theorem}
To prove this theorem, we need the following identities and estimate. They are nothing but 
elaborations of Theorem~\ref{123}(2), Lemma~\ref{duhammel}, 
and Lemma~\ref{getzler}.

\begin{definition}
Let $V$ and $W$ be operators affiliated with $\Nn$ such that
they have the same degree as $\Dd$, i.e.
 $|\Dd|_\chi = |V|_\chi =|W|_\chi$.
Define $\hCh^\bullet(\Dd,V,W)$ to be given by the equation
\begin{eqnarray*}
\lefteqn{\left(\hCh^n(\Dd,V,W),(a_0,\ldots,a_n)_n\right)}
\\&:=&
 \sum_{k=1}^{j} (-1)^{k} \sum_{j=1}^{n+1}(-1)^j 
\left\langle a_0, \ldots,[\Dd,a_{k-1}],W,\ldots, V ,[\Dd,a_j]\ldots 
\right\rangle^{n+2}_{\Dd} 
\\&&+
 \sum_{k=j}^{n+2} (-1)^{k+1} \sum_{j=1}^{n+1}(-1)^j 
\left\langle a_0, \ldots, [\Dd,a_{j-1}],V,\ldots ,W,[\Dd,a_k],\ldots
\right\rangle^{n+2}_{\Dd}  \mbox{ .}
\end{eqnarray*}
\end{definition}

\begin{lemma}
Let $V$ and $W$ be operators affiliated with $\Nn$. Then we have
\label{level2aux}
 \begin{eqnarray*}
\lefteqn{ b \hCh^{n-1}(\Dd,V,W)+B  \hCh^{n+1}(\Dd,V,W) }\\
 &=& \iota(V)\left( \iota(\Dd W+W \Dd) \Ch^n(\Dd) 
 -\alpha^n(\Dd,W)\right)
\\&& -
 \iota(W)\left(\iota(\Dd V+V\Dd)\Ch^n(\Dd)
- \alpha^n(\Dd,V)\right) \mbox{ .}
 \end{eqnarray*}
\end{lemma}
The above lemma can be found in \cite{chernreduct}.
Its proof is nothing but an elaboration of the proof to
 Lemma~\ref{123}(2), which is a  lengthy but straight forward calculation,
 so we decide to skip it here.
\begin{lemma}
\label{derivative}
Suppose that $\Dd_s$ and $V_s$ are 1-parameter
 families of operators affiliated with $\Nn$ so that $\dot{\Dd_s}$ and $\dot{V_s}$ are defined
and affiliated with $\Nn$, then
 \[
  \frac{d }{ds}\hCh^n(\Dd_s,V_s)
=\hCh^n(\Dd_s,\dot{V_s})+\iota(V_s)\alpha^n(\Dd_s,\dot{\Dd_s})
-\iota(V_s) \iota(\Dd_s \dot{\Dd_s} + \dot{\Dd_s} \Dd_s)
\Ch ^n(\Dd_s ) 
\mbox{ .}
 \]
\end{lemma}
\begin{proof}
By applying Leibniz rule on $\frac{d }{ds}\hCh^n(\Dd_s,V_s)$,
we will obtain a sum of the term containing $\frac{d }{ds}V_s$,
terms containing
the $\frac{d }{ds}[\Dd_s,a_k]$ , and the  terms containing
$\frac{d }{ds} e^{-s_j\Dd_s^2}$.
They collect into 
 $\hCh^n(\Dd_s,\dot{V_s})$, $\iota(V_s)\alpha^n(\Dd_s,\dot{\Dd_s})
$, and by Lemma~\ref{duhammel}, 
$-\iota(V_s) \iota(\Dd_s \dot{\Dd_s} + \dot{\Dd_s} \Dd_s)
\Ch ^n(\Dd_s ) $.
\end{proof}

By using Lemma \ref{level2aux} and \ref{derivative},
we can establish the algebraic equality
\[
\frac{d\cch^n_t(\Dd_\alpha)}{d\alpha} = (b+B)\left(
\int^t_0 
b\iota(\Dd_\alpha)\hCh^{n}(u\Dd_\alpha,u\dot{\Dd_\alpha})du -\hCh^{\leq n-1}(t\Dd_\alpha,t\dot{\Dd_\alpha}) \right)
 \mbox{ .}
\]
However, the major task is to show that
in fact 
\[
\int^t_0 
b\iota(\Dd_\alpha)\hCh^{n}(u\Dd_\alpha,u\dot{\Dd_\alpha})du -\hCh^{\leq n-1}(t\Dd_\alpha,t\dot{\Dd_\alpha}) \] is entire.
And the analysis required to prove its entireness is a little involved.

We begin by observing the operator 
$[F \ln \lvert \Dd \rvert , a]$ is bounded.
 To show this, we need the following lemma.
\begin{lemma}
\label{somefunctionalequation}
Let $H$ be a positive operator. Then
\begin{eqnarray*}
H^{-\frac{\alpha}{2}} \ln H &=& \frac{1}{C_\alpha }
\int _0 ^\infty (H+\lambda)^{-1} \lambda ^{-\frac{\alpha}{2}}  
(\ln \lambda )
d\lambda -   \frac{C_\alpha '}{C_\alpha }H^{-\frac{\alpha}{2}} 
\end{eqnarray*} for $\alpha >0$, where
$C_\alpha=\int _0 ^\infty (1+x)^{-1} x^{-\alpha/2}  dx$ and 
$C_\alpha '=\int _0 ^\infty (1+x)^{-1} x^{-\alpha/2} ( \ln x )dx$.
\end{lemma}
\begin{proof}
From changing variable $x=\lambda/y$, we obtain 
\begin{eqnarray*}
\int _0 ^\infty (1+x)^{-1} x^{-\alpha/2} dx &=&
y^{\alpha/2} \int _0 ^\infty (y+\lambda)^{-1} \lambda ^{-\alpha/2} 
d\lambda \mbox{ .}
\end{eqnarray*}
By differentiating both sides with respect to $\alpha$,
the above turns into
\begin{eqnarray*}
\int _0 ^\infty (1+x)^{-1} x^{-\alpha/2} ( \ln x )
 dx &=&
y^{\alpha/2} \int _0 ^\infty (y+\lambda)^{-1} \lambda ^{-\alpha/2} 
\ln (\lambda / y)
d\lambda\\ &=&
y^{\alpha/2} \int _0 ^\infty (y+\lambda)^{-1} \lambda ^{-\alpha/2} 
(\ln \lambda )
d\lambda\\
&&-
\ln y
 \int _0 ^\infty (1+x)^{-1} x ^{-\alpha/2} 
dx \mbox{ ,}
\end{eqnarray*}
where the integrals converge as long as $\alpha >0$.
Now using functional calculus to substitute $H$ in $y$ to get
\begin{eqnarray*}
H^{-\frac{\alpha}{2}} \ln H &=& \frac{1}{C_\alpha }
\int _0 ^\infty (H+\lambda)^{-1} \lambda ^{-\frac{\alpha}{2}}  
(\ln \lambda )
d\lambda -   \frac{C_\alpha '}{C_\alpha }H^{-\frac{\alpha}{2}} \mbox{ ,}
\end{eqnarray*}
which is the desired equation.
\end{proof}

\begin{proposition}
\label{someshitisbounded}
Let $\Dd$ be invertible and $F=\Dd \lvert \Dd \rvert ^{-1} $.
For any $a\in A$, the commutator
$[ F \ln \lvert \Dd \rvert ,a ]$ is bounded.
\end{proposition}
\begin{proof}
By applying Lemma~\ref{somefunctionalequation} for $H=\Dd^2$ and
$\alpha=1$, one obtains
\begin{eqnarray*}
2[F
 \ln \lvert \Dd \rvert , b]
&=&
[\Dd,b] \lvert \Dd \rvert ^{-1} \ln \Dd^2 
+
\Dd [\lvert \Dd \rvert ^{-1} \ln \Dd^2 ,b]\\
&=&
[\Dd,b]\left(
\frac{1}{C_1} \int_0 ^\infty (\Dd^2+\lambda)^{-1} 
\lambda^{-1/2}(\ln \lambda) d\lambda - \frac{C_1 '}{C_1} \lvert \Dd \rvert ^{-1}
\right)\\
&&+
\Dd \left(
\frac{1}{C_1} \int_0 ^\infty
[ (\Dd^2+\lambda)^{-1} ,b]
\lambda^{-1/2}(\ln \lambda) d\lambda -
 \frac{C_1 '}{C_1}[ \lvert \Dd \rvert ^{-1},b]
\right)\\
&=& \frac{1}{C_1} \int _0 ^\infty 
[ \Dd (\Dd^2 + \lambda)^{-1},b] \lambda ^{-1/2}(  \ln \lambda)
d\lambda
- \frac{C_1 '}{C_1} [F,b]  
\end{eqnarray*}
for $b \in A$.
Since $[F,b]$ is bounded, we see that $[F\ln \lvert \Dd \rvert,b]$
is bounded  if and only if 
$\frac{1}{C_1} \int _0 ^\infty 
[ \Dd (\Dd^2 + \lambda)^{-1},b] \lambda ^{-1/2}(  \ln \lambda)
d\lambda$ is bounded. We compute
\begin{eqnarray*}
\lefteqn{ \frac{1}{C_1} \int _0 ^\infty 
[ \Dd (\Dd^2 + \lambda)^{-1},b] \lambda ^{-1/2}(  \ln \lambda)
d\lambda} \\ &=&
\frac{1}{C_1} \int _0 ^\infty  \left(
[ \Dd,b]  (\Dd^2 + \lambda)^{-1}+
\Dd[ (\Dd^2 + \lambda)^{-1},b] \right)
\lambda ^{-1/2}(  \ln \lambda) d\lambda \\
&=&
\frac{1}{C_1} \int _0 ^\infty 
\left( [ \Dd,b]  (\Dd^2 + \lambda)^{-1}
\right. \\ && \left.
-
\Dd (\Dd^2 + \lambda)^{-1}
 [ \Dd^2 + \lambda,b] 
(\Dd^2 + \lambda)^{-1}
\right)
\lambda ^{-1/2}(  \ln \lambda) d\lambda \\
&=&
\frac{1}{C_1} \int _0 ^\infty 
\left(
(\Dd^2 + \lambda)(\Dd^2 + \lambda)^{-1}[ \Dd,b]  (\Dd^2 + \lambda)^{-1}
\right. \\ && \left.
-
\Dd (\Dd^2 + \lambda)^{-1} 
\left(\Dd [ \Dd  ,b] + [ \Dd ,b] \Dd\right)
 (\Dd^2 + \lambda)^{-1} \right)
\lambda ^{-1/2}(  \ln \lambda) d\lambda \\
&=&
\frac{1}{C_1} \int _0 ^\infty 
\left(
 \lambda(\Dd^2 + \lambda)^{-1}[ \Dd,b]  (\Dd^2 + \lambda)^{-1}
\right. \\ && \left.
-
\Dd (\Dd^2 + \lambda)^{-1} 
 [ \Dd ,b] \Dd
 (\Dd^2 + \lambda)^{-1} \right)
\lambda ^{-1/2}(  \ln \lambda) d\lambda  \mbox{ .}
\end{eqnarray*}
If $[D,b]$ is self-adjoint, then the estimate 
$  -\lVert [\Dd,b] \rVert  \leq   [\Dd,b]  \leq \lVert [\Dd,b] \rVert$ 
and the fact that $\Dd$ is self-adjoint yield
\begin{eqnarray*}
- \lambda  (\lambda +\Dd^2)^{-1}\left\lVert [\Dd,b] \right\rVert  (\lambda  +\Dd^2)^{-1}  &\leq &
\lambda (\lambda+\Dd^2)^{-1}[\Dd,b]  (\lambda +\Dd^2)^{-1} \\
&\leq & \lambda  (\lambda  +\Dd^2)^{-1}\lVert[\Dd,b]\rVert(\lambda +\Dd^2)^{-1} 
\mbox{ ,}
\end{eqnarray*}
and
\begin{eqnarray*}
  -|\Dd|(\lambda  +\Dd^2)^{-1}\lVert [\Dd,b] \rVert  (\lambda  +\Dd^2)^{-1}|\Dd| &\leq & 
-\Dd(\lambda  +\Dd^2)^{-1}[\Dd,b]  (\lambda  +\Dd^2)^{-1}\Dd \\
&\leq&  |\Dd|(\lambda +\Dd^2)^{-1}\left\lVert[\Dd,b]\right\rVert(\lambda  +\Dd^2)^{-1}|\Dd| 
\mbox{ .}
\end{eqnarray*}
Hence,
\begin{eqnarray*}
   -\left\lVert [\Dd,b] \right\rVert 
\left(
 | \Dd|^{-1} \ln \lvert \Dd \rvert + \frac{C_1 ' }{C_1} \lvert \Dd \rvert ^{-1} \right)
&\leq&
\frac{1}{C_1} \int _0 ^\infty 
[ \Dd (\Dd^2 + \lambda)^{-1},b] \lambda ^{-1/2}(  \ln \lambda)
d\lambda \\
&\leq  &
\left\lVert [\Dd,b] \right\rVert \left(
 | \Dd|^{-1} \ln \lvert \Dd \rvert + \frac{C_1 ' }{C_1} \lvert \Dd \rvert ^{-1} \right)
\mbox{ .}
\end{eqnarray*} 
Therefore, in the end we obtain
\begin{eqnarray}
\label{unknownequation1}
\left \lVert
[F  \ln \lvert \Dd \rvert , b]
\right \rVert
&\leq&\frac{1}{2}\left\lVert  [\Dd,b] \right\rVert  
\left( 
\left\lVert | \Dd|^{-1} \ln \lvert \Dd \rvert \right\rVert
+ \frac{C_1 '}{C_1} \left\lVert  \lvert \Dd \rvert ^{-1} \right \rVert 
\right)
+ \frac{C_1'}{2C_1} \left\lVert [F,b]\right\rVert
 \mbox{ ,}
\end{eqnarray} which is bounded.
Since for any $a\in A$, 
$[F\ln \lvert \Dd \rvert ,a]=[F\ln \lvert \Dd \rvert 
,\frac{a-a^*}{2}]+i[F\ln \lvert \Dd \rvert ,\frac{a+a^*}{2i}]$ with
 $[F\ln \lvert \Dd \rvert ,\frac{a-a^*}{2}]$ and 
$[F\ln \lvert \Dd \rvert ,\frac{a+a^*}{2i}]$ self-adjoint.
By Equation~\ref{unknownequation1},
$[F\ln \lvert \Dd \rvert ,\frac{a-a^*}{2}],[F\ln \lvert \Dd \rvert ,\frac{a+a^*}{2i}]$ are bounded.
Hence, so is  $
\left \lVert [F\ln \lvert \Dd \rvert ,a]
\right \rVert $.
\end{proof}

\begin{proposition}
\label{somethingannoyingisentire}
  For a $p$-summable unbounded Breuer-Fredholm module $\KCYCLE$ with $\Dd$ invertible, 
set $\Dd_\alpha = \Dd |\Dd|^{-\alpha}$.
 Then
$\int_0 ^t b \iota(\Dd_\alpha )
\hCh^n(u\Dd_\alpha,u\dot{\Dd_\alpha})du$  for $n-1 >p$,
is a well-defined family of entire cyclic cocycles for $\alpha\in[0,1]$
and  $t\in [1,\infty]$,
where $\dot{\Dd}_\alpha =- \Dd_\alpha \ln \lvert \Dd \rvert$.
\end{proposition}
\begin{proof} 
The proof will go in two steps. First we 
estimate the norm of a generic term of $\int_0^t
\iota(\Dd_\alpha)
\hCh^{n}(u\Dd_\alpha,u\Dd_\alpha \ln \lvert \Dd \rvert)du $
for $\alpha \in [0,1)$. 
Set $\lambda=\inf \left(\sigma(\Dd^2)\right) >0$, we compute:
\begin{eqnarray*}
\lefteqn{ 
\left \lVert 
\left\langle 
b_0,[u\Dd_\alpha,b_1],\ldots,
u\Dd_\alpha \ln \lvert \Dd \rvert,  \ldots,
\Dd_\alpha ,\ldots
[u\Dd_\alpha,b_{n}]
\right\rangle 
_{u\Dd_\alpha}^{n+2}
\right \rVert
} \\
&\leq &
\int_0 ^t u^{n+1}
\int_{\Delta_{n+2}}\left\lVert b_0 e^{-s_{0}(u\Dd_\alpha)^2}
[\Dd_\alpha,b_1]e^{-s_{1}(u\Dd_\alpha)^2} \cdots 
\Dd_\alpha \ln \lvert \Dd \rvert 
e^{-s_{k}(u\Dd_\alpha)^2 } \cdots
 \right.
\\ && \hspace{1cm}\left.
\cdots
  \Dd_\alpha    e^{-s_{n+2}(u\Dd_\alpha)^2}   \cdots
[\Dd_\alpha,b_{n}]e^{-s_{n+2}(u\Dd_\alpha)^2 }
   \right\rVert_1 d^{n+2}s   \mbox{ }du \\
& \stackrel{\ref{holder} }{\leq }& 
\left\lVert b_0 \right \rVert
\left\lVert [\Dd_\alpha,b_1]
 \right\rVert_{\frac{ \lceil p\rceil}{\alpha}} \cdots
\left\lVert [\Dd_\alpha,b_{\lceil p\rceil}]
 \right\rVert_{\frac{\lceil p\rceil}{\alpha}} 
 \left\lVert [\Dd_\alpha,b_{\lceil p\rceil+1}] \right\rVert  \cdots
\left\lVert [\Dd_\alpha,b_{n}] \right\rVert 
\\&& \int_{0}^t u^{n+1}
\int_{\Delta_{n+2}} 
\left( 
\left\lVert e^{-s_{0}(u\Dd_\alpha)^2} 
\right\rVert_{\frac{1}{(1-\alpha)s_{0}}}
\cdots 
\left\lVert   \Dd_\alpha \ln \lvert \Dd \rvert
 e^{-   s_{k}(u\Dd_\alpha)^2}
  \right\rVert _{\frac{1}{(1-\alpha)s_{k}}}
 \cdots \right.
 \\
 && \hspace{2cm} \left. \cdots
\left\lVert \Dd_\alpha e^{- s_{j}(u\Dd_\alpha)^2 }
\right\rVert_{\frac{1}{(1-\alpha)s_{j}}}\cdots
\left\lVert e^{-s_{n+2}(u\Dd_\alpha)^2 }
\right\rVert_{\frac{1}{(1-\alpha)s_{n+2}}} 
 \right) d^{n+2}s \mbox{ } du \\ 
&\stackrel{\ref{pp}}{\leq}& \left \lVert b_0 \right \rVert
\left( 
\prod _{j=1}^{n}
\left\lVert [\Dd ,b_j]  \right\rVert \right) 
\left( \left \lVert \lvert \Dd  \rvert 
^{-1} \right\rVert _p \right)^{p \alpha}
 \int_{0}^t u^{n+1}
\int_{\Delta_{n+2}} \left( 
\left\lVert e^{-s_{0}(1-\delta)(u\Dd_\alpha)^2} 
\right\rVert_{\frac{1}{(1-\alpha)s_{0}}}
\cdots \right.
 \\ && \left.  \hspace{2cm}
\left\lVert e^{-(1-\delta)s_{n+2}(u\Dd_\alpha)^2 }\right\rVert_{\frac{1}{(1-\alpha)s_{n+2}}}
 \left\lVert \lvert \Dd_\alpha \rvert^{-\varepsilon}\ln \lvert \Dd \rvert
\right\rVert
\right.
 \\
 &&  \left. \hspace{3cm}
 \left\lVert \lvert \Dd_\alpha \rvert^{1+\varepsilon}
 e^{- \delta s_{k}(u\Dd_\alpha)^2}
  \right\rVert 
 \left\lVert \lvert \Dd_\alpha \rvert 
e^{- \delta s_{j}(u\Dd_\alpha)^2}
  \right\rVert 
 \right) d^{n+2}s \mbox{ } du \\
\end{eqnarray*}
\begin{eqnarray*}
&\leq& C^n
\left( \prod _{j=1}^{n}
\left\lVert b_j  \right\rVert \right) 
\left( \left \lVert \lvert \Dd \rvert 
^{-1} \right\rVert _p \right)^{p \alpha}
\left \lVert ( x^{1-\alpha})^{-\varepsilon} 
\ln x  \right \rVert _\infty
 \int_{0}^t u^{n+1}
\left\lVert e^{-(1-\delta)(u\Dd_\alpha)^2} 
\right\rVert_{\frac{1}{(1-\alpha)}}
 \\ && \hspace{2cm}
\int_{\Delta_{n+2}} \left( 
\left \lVert  x^{1+\varepsilon}  e ^{-\delta s_{k} (ux)^2
}  \right \rVert _\infty
\left \lVert   x e ^{-\delta s_{j} (ux)^2}
  \right \rVert _\infty
\right)d^{n+2}s   \mbox{ } du \\
&
\leq&
C^n
\left( \prod _{j=1}^{n}
\left\lVert b_j  \right\rVert \right) 
\frac{    \left( \left \lVert \lvert \Dd  \rvert 
^{-1} \right\rVert _p \right)^{p \alpha}   }{e\varepsilon(1-\alpha)}
 \int_{0}^t u^{n+1} \left 
\lVert e^{-(1-\delta)(u\Dd_\alpha)^2/2} \right \rVert
\left\lVert e^{-(1-\delta)(u\Dd_\alpha)^2/2} 
\right\rVert_{\frac{1}{(1-\alpha)}}
 \\ && \hspace{2cm}
\int_{\Delta_{n+2}} \left( 
\left(\frac{1+\varepsilon}
{2e\delta s_{k}u^2}\right)^{\frac{1+\varepsilon}{2}}
\left(\frac{1 }
{2e\delta s_{j}u^2}\right)^{\frac{1 }{2}}
\right)d^{n+2}s   \mbox{ }du
\\
&\leq&
C^n
\left( \prod _{j=1}^{n}
\left\lVert b_j  \right\rVert \right) 
\frac{  4  \left( \left \lVert \lvert \Dd  \rvert 
^{-1} \right\rVert _p \right)^{p \alpha} }{e\varepsilon(1-\alpha)n!}
\left( \frac{1}{2e\delta}
\right) ^{1+\frac{\varepsilon}{2}} 
\frac{(1+\varepsilon)^{\frac{1+\varepsilon}{2}} }{1-\varepsilon}
 \\
&& \hspace{2cm}
\left\lVert
\lvert \Dd_\alpha \rvert ^{-p}
\right\rVert_{\frac{1}{(1-\alpha)}}
 \int_{0}^t u^{n-1-\varepsilon}
 e^{-(1-\delta)\lambda^{1-\alpha}u^2/2} 
\left\lVert \lvert \Dd_\alpha \rvert ^{p} 
e^{-(1-\delta)(u\Dd_\alpha)^2/2} 
\right\rVert du\\
&\leq&
\frac{C^n
 \prod _{j=1}^{n}
\left\lVert b_j  \right\rVert 
    \left( \left \lVert \lvert \Dd  \rvert 
^{-1} \right\rVert _p \right)^{p \alpha} (1+\varepsilon)^{\frac{1+\varepsilon}{2}}
\left\lVert
\lvert \Dd_\alpha \rvert ^{-p}
\right\rVert_{\frac{1}{(1-\alpha)}} }
{(1-\alpha)2^{\frac{p+\varepsilon-2}{2}}
p^{\frac{-p}{2}}\delta^{1+\frac{\varepsilon}
{2}}(1-\delta)^{\frac{p}{2}}\varepsilon(1-\varepsilon)
e^{\frac{p+4+\varepsilon}{2}}n!} 
 \int_{0}^t u^{n-p-1-\varepsilon}
 e^{-(1-\delta)\lambda^{1-\alpha}u^2/2} 
 du
\mbox{ .} 
\end{eqnarray*}
The integral
\begin{equation}
\label{eqn-ch4-someintegral}
\int_{0}^t u^{n-p-1-\varepsilon}
 e^{-(1-\delta)\lambda^{1-\alpha}u^2/2} 
 du\end{equation} 
exists for $t\in[1,\infty]$ as long as $n-1-\varepsilon \geq p$.
Thus,
$b\int_0^\infty
\iota(\Dd_\alpha  )
\hCh^{n}(u\Dd_\alpha,u\Dd_\alpha \ln \lvert \Dd \rvert)du $
 is entire
for $t\in[1,\infty]$, $\alpha \in [0,1)$,  
$n-1 > p$, and $0<\delta,\varepsilon$ sufficiently small.
%
%
Now we suppose that $\alpha=1$.
Since
 \[
\int_0^t \iota( F  )
\hCh^n (uF,u F \ln \lvert \Dd \rvert ) du=
\iota(F)
 \hCh^n (F,F\ln \lvert \Dd \rvert)
\int_0^t u^{n+1} e^{-u^2}du \mbox{ ,}
\] to show 
$\int_0^t b\iota(F )
\hCh^n (uF,u F \ln \lvert \Dd \rvert) du$ is bounded, 
 it suffices to know that $b \hCh^n (F,F\ln \lvert \Dd \rvert)$
is bounded. Hence,
we estimate the norm of a generic term of
$
\hCh^n (F,F\ln \lvert \Dd \rvert) $ paired with $b(a_0,\ldots,a_{n+1})_{n+1}$.
\begin{eqnarray*}
 \lefteqn{ \left \lVert  \sum_{j=0}^{n+1} (-1)^j   
\tau\left(\chi
a_0 \cdots [F,a_j a_{j+1}] \cdots F \ln \lvert \Dd \rvert  
\cdots [F,a_{n+1}]
\right)  \right \rVert }\\
&&\leq  
\left \lVert 
\tau\left(\chi a_0 [F,a_1] \cdots  [F\ln\lvert \Dd \rvert,a_{k}]
 \cdots [F,a_{n+1}] \right) \right \rVert \\
&&\leq
 \left \lVert [F\ln \lvert \Dd \rvert ,a_k] \right \rVert
\left \lVert a_0 \right \rVert
\prod _{\stackrel{j=1}{j\neq k}}  ^{n+1}
\left \lVert [F,a_j] \right \rVert _{n} \mbox{ ,}
\end{eqnarray*}
which is bounded by Proposition~\ref{someshitisbounded}.
Therefore, by the continuity,
\[\int_0 ^t b \iota(\Dd_\alpha )
\hCh^n(u\Dd_\alpha,u \dot{\Dd}_\alpha)du\] is entire for $\alpha\in [0,1]$.
\end{proof}

The same proof will show that 
$\hCh^{\leq n-1}(t\Dd_\alpha,t\dot{\Dd_\alpha})$ is entire.
Since we do not need to integrate with respect to $u$,
 a term like Equation~\eqref{eqn-ch4-someintegral} does not appear.
There does not need to be a lower bound for $n$,
so by the same estimate techniques deployed in the 
proof of Proposition~\ref{somethingannoyingisentire}, it is easy to see that
$\hCh^{\leq n-1}(t\Dd_\alpha,t\dot{\Dd_\alpha})$ is entire.

\begin{proof}[Proof of Theorem~\ref{cchisexact}]

We compute
  \begin{eqnarray*}
  \frac{d\cch^n_t(\Dd_\alpha)}{d\alpha}
 &=&
\frac{d }{d\alpha} \Ch^{\leq n} (t \Dd_\alpha) + B\int^t_0 \frac{d }{d\alpha} \hCh^{n+1}(u\Dd_\alpha,\Dd_\alpha)du \\
&\stackrel{\ref{cobound}}{=}&
-(b+B)\hCh^{\leq n-1}(t\Dd_\alpha,t\dot{\Dd_\alpha}) - B\hCh^{n+1}(t\Dd_\alpha,t\dot{\Dd_\alpha})
\\&& \hspace{2cm} +
B\int^t_0 \frac{d }{d\alpha} \hCh^{n+1}(u\Dd_\alpha,\Dd_\alpha)du \\
&=&-(b+B)\hCh^{\leq n-1}(t\Dd_\alpha,t\dot{\Dd_\alpha})
\\&& \hspace{2cm} +B\int^t_0 \left( \frac{d }{d\alpha}
 \hCh^{n+1}(u\Dd_\alpha,\Dd_\alpha) 
-\frac{d }{du}\hCh^{n+1}(u\Dd_\alpha,u\dot{\Dd_\alpha}) \right) du\\
&\stackrel{\ref{derivative}}{=}&
-(b+B)\hCh^{\leq n-1}(t\Dd_\alpha,t\dot{\Dd_\alpha}) \\
&&+ B\int^t_0 \left( 
\hCh^{n+1}(u\Dd_\alpha,\dot{\Dd_\alpha}) +\iota(\Dd_\alpha)\alpha^{n+1}(u\Dd_\alpha,u\dot{\Dd_\alpha})
\right. 
\\&& \hspace{2cm} \left. - \iota(\Dd_\alpha)
\iota(u\Dd_\alpha u\dot{\Dd_\alpha} + u\dot{\Dd_\alpha}u\Dd_\alpha)
 \hCh^{n+1}(u\Dd_\alpha) \right. \\
 && \hspace{3cm} - \left. \hCh^{n+1}(u\Dd_\alpha,\dot{\Dd_\alpha}) - \iota(u\dot{\Dd_\alpha})\alpha^{n+1}(u\Dd_\alpha,\Dd_\alpha)
\right. 
\\&& \hspace{4cm} \left.
+\iota(u\dot{\Dd_\alpha}) 
\iota(u\Dd_\alpha\cdot \Dd_\alpha + \Dd_\alpha \cdot u\Dd_\alpha)
\hCh^{n+1}
(u\Dd_\alpha)
\right) du 
\end{eqnarray*}\begin{eqnarray*}
&\stackrel{\ref{level2aux}}{=}&
-(b+B)\hCh^{\leq n-1}(t\Dd_\alpha,t\dot{\Dd_\alpha}) + 
B\int^t_0 \left( 
b\iota(\Dd_\alpha)\hCh^{n}(u\Dd_\alpha,u\dot{\Dd_\alpha}) 
\right.
\\ && \hspace{4cm} \left. +B\iota(\Dd_\alpha)
\hCh^{n+2}(u\Dd_\alpha,u\dot{\Dd_\alpha})
\right) du \\
&=&  (b+B) \left(\int^t_0 
b\iota(\Dd_\alpha)
\hCh^{n}(u\Dd_\alpha,u\dot{\Dd_\alpha})du -\hCh^{\leq n-1}(t\Dd_\alpha,t\dot{\Dd_\alpha})
\right) \mbox{ ,}
 \end{eqnarray*}
where the last equality follows from the identities
 $b^2=B^2=0$.

Finally, 
$\int^t_0 
b\iota(\Dd_\alpha)\hCh^{n}(u\Dd_\alpha,u\dot{\Dd_\alpha})du$
is entire by Proposition~\ref{somethingannoyingisentire}.
$\hCh^{\leq n-1}(t\Dd_\alpha,t\dot{\Dd_\alpha})$ is also
entire by the same proof as Proposition~\ref{somethingannoyingisentire}.
%
\end{proof}

\begin{proposition}
  For a $p$-summable unbounded Breuer-Fredholm module $\KCYCLE$ with $\Dd$ invertible, 
set $\Dd_\alpha = \Dd |\Dd|^{-\alpha}$. Then
$\lim _{t\rightarrow \infty} \cch^n_t(\Dd_\alpha)$ is
 a  family of entire cocycles for $\alpha\in[0,1]$.
\end{proposition}
\begin{proof}
Same proof as Proposition~\ref{somethingannoyingisentire}.
\end{proof}

\begin{theorem}
\label{done}
 For a $p$-summable unbounded Breuer-Fredholm module $\KCYCLE$ with
 $\Dd$ invertible,
 its JLO character is cohomologous to the Connes character of its associated Breuer-Fredholm module.
\end{theorem}
\begin{proof}
By Theorem~\ref{cchisexact},
\[
B\int_0^\infty\hCh^{n+1}(uF,F)du -
\lim _{t\rightarrow \infty}\cch^n_t(\Dd ) 
 =(b+B) \int_0 ^1 \int^\infty _0 
b\iota(u\dot{\Dd_\alpha})\hCh^{n}(u\Dd_\alpha,\Dd_\alpha)du \mbox{ } d\alpha
\mbox{ .}
\]
Together with Proposition~\ref{finitejlo}, we conclude that
$B\int_0^\infty\hCh^{n+1}(uF,F)du$ is cohomologous to the JLO character
$\Ch^\bullet(\Dd)$, where $F=\Dd |\Dd|^{-1}$.

The map $B: C_n(\Bb)\rightarrow  C_{n+1}(\Bb)$ on chains
 can be decomposed into $B=sN$ where 
\begin{eqnarray*}
 N(a_0,\ldots,a_n)&:=&\sum^n_{j=0}(-1)^{nj}(a_j,\ldots,a_n,a_0,\ldots,a_{j-1})  \\
s(a_0,\ldots,a_n)&:=& (1,a_0,\ldots,a_n) \mbox{ .}
\end{eqnarray*}

By observing the fact that $F[F,a]=-[F,a]F$ and combining
 Lemma~\ref{misc}(1)(2),
the rest is straightforward calculation:
\begin{eqnarray*}
\lefteqn{ \left( B\int^\infty_0\hCh^{n+1}(uF,F)du , (a_0,\ldots, a_n) \right) }
 \\ &=&
 N\int^\infty_0 \langle F, [uF,a_0],\ldots,[uF,a_n] \rangle ^{n+1}_{uF} du \\
&=& \left(\int^\infty_0 u^{n+1} e^{-u^2} du \int_{\Delta_{n+1}}ds\right) N \tau 
\left( \chi F[F,a_0]\cdots[F,a_n] \right) \\
&=& \left(\frac{1}{2}\int^\infty_0 t^{n/2} e^{-t} dt \frac{1}{(n+1)!}\right) (n+1) \tau 
\left( \chi F[F,a_0]\cdots[F,a_n] \right) \\
&=& \frac{\Gamma{(\frac{n}{2}+1})}{2 \cdot n!} \tau 
\left( \chi F[F,a_0]\cdots[F,a_n] \right) \\
&=& \left( \ch^n(F), (a_0,\ldots, a_n) \right)
\mbox{ .}
\end{eqnarray*}
\end{proof}


The invertibility assumption in this section
can be removed  as follows (see \cite{cprs4}).
Given an unbounded Breuer-Fredholm module $\KCYCLE$, 
we can associate to it another unbounded Breuer-Fredholm module $(\rho',\Nn',\Dd')$ 
with $\Dd'$ invertible.
 First we form the sum $\KCYCLE \oplus (0,\Nn,-\Dd)
:=(\rho\oplus 0, \Nn\otimes M_2(\mathbb{C}) , \Dd \oplus -\Dd)$ and equip it with the grading 
$\chi \oplus -\chi$,
then perturb $\Dd \oplus -\Dd$ by the  isometry  
\[
\left(
\begin{array}{cc}
  0 & K \\ K & 0
 \end{array}\right) \in \Nn\otimes M_2(\mathbb{C})\] 
 that exchanges the two copies of $\Hh$.
Here $K$ is made to be
\textit{odd} with respect to the grading by exchanging the $\Hh^+$ and $\Hh^-$ subspaces when $\Dd$ is graded, 
\[
 K=\left(
\begin{array}{cc}
  0 & 1 \\ 1 & 0
 \end{array}\right) \mbox{ .}
\]
We set
\[ (\rho',\Nn',\Dd'):=
\left(\rho\oplus 0 , \Nn \otimes  M_2(\mathbb{C}), 
 \left(
\begin{array}{cc}
  \Dd & K \\ K & -\Dd
 \end{array}\right) \right)\mbox{ ,}
\]
then $(\rho',\Nn',\Dd')$ is an unbounded Breuer-Fredholm module,
 and $\Dd'$ has 
the same summability as $\Dd$.
The identity
\begin{eqnarray*}
\Dd'^2  = \left(
\begin{array}{cc}
  \Dd & K \\ K & -\Dd
 \end{array}\right) ^2
= \left( \begin{array}{cc}
  \Dd^2+1 & 0 \\ 0 & \Dd^2+1
 \end{array}\right) 
= \left(
\begin{array}{cc}
  (\Dd+i)(\Dd-i) & 0 \\ 0 & (\Dd+i)(\Dd-i)
 \end{array}\right)
\end{eqnarray*}
implies that $\Dd'$ is a bijection (from its domain), and is invertible.
 Furthermore, it represents the same $\KKEI$-homology class
as $\KCYCLE$. 
As the procedure of obtaining $(\rho',\Nn',\Dd')$ can be
 described by
adding zero to $\KCYCLE$  and perturbing the sum by an isometry, 
which is the equivalence relation in $\KKEI$-homology.

\appendix
\section{Appendix}
The Appendix gives an account on basic definitions needed for the discussion of the paper,
it includes affiliated operators, $\tau$-compact operators, $p$-summable operators etc.,
 then followed by some basic properties of these operators.
The ideals $\Kk$ and $\Ll^p$   are then defined in terms of  $\tau$-compactness and $p$-summability,
and finally the Appendix ends by stating H\"{o}lder's inequality, which is crucial in our work.

The presentation in this section follows closely 
 \cite{snumber}
 and \cite{type2index}
to which we refer to
proofs and further details.

A von Neumann algebra with underlying Hilbert space $\Hh$ is
a unital $*$-subalgebra of the algebra of bounded operators $B(\Hh)$
on $\Hh$ that is closed under the weak operator topology.

A positive linear functional on a von Neumann algebra is said to be
 \textbf{normal} if it preserves $\sup$ of any increasing nets of
positive operators in the von Neumann algebra;
\textbf{faithful} if it is positive-definite on positive operators;
\textbf{semi-finite} if the $*$-subalgebra generated by positive
elements with finite value under the functional is $\sigma$-weak
dense in the von Neumann algebra \cite{OAQSM}.
A von Neumann algebra is called \textbf{semi-finite} if 
it admits a faithful, semi-finite normal trace.
A von Neumann algebra is \textbf{Type I} if it is
semi-finite and every projection contains a minimal sub-projection;
\textbf{Type II} if it is semi-finite but not Type I \cite{OAQSM}.

Let $\Nn$  be a semi-finite von Neumann algebra
with underlying Hilbert space $\Hh$ and a faithful semi-finite normal trace
 $\tau$.

\begin{definition}
 A densely defined closed operator $T$ on $\Hh$  with polar decomposition $T=U|T|$ \cite{reedsimon}
is said to be \textbf{affiliated} with $\Nn$ if
$U\in \Nn$ and also 
the spectral projection $1_{[0,\lambda]}(\lvert T \rvert )$
of $|T|$ 
 lies in $\Nn$ for all $\lambda$, where 
$1_{[0,\lambda]}$ is the characteristic function supported on the
closed interval $[0,\lambda]\in \mathbb{R}$.
\end{definition}

  For a positive self-adjoint operator
 $T=\int_0^\infty \lambda dE_\lambda$ affiliated with $\Nn$
with $E_\lambda=1_{[0,\lambda]}(\lvert T\rvert )$,
 we define its semi-finite trace  by
\[
 \tau(T)=\int_0^\infty \lambda d \tau(E_\lambda) \mbox{ .}
\]

 From now on, when we say that an operator $T$ is affiliated with $\Nn$, 
we implicitly demand that $T$ is densely defined and closed.

\begin{definition}
For an operator $T$ affiliated with $\Nn$ and $x>0$, 
the \textbf{generalized singular number} $\mu_x(T)$ 
with respect to $(\Nn,\tau)$
 is defined to be
\[
 \mu_x(T):=\inf_{E} \left\{\left\lVert T E\right\rVert: \tau(1-E) 
\leq x \right\} \mbox{ ,}
\]
where the infimum  is taken over projections $E\in \Nn$.
\end{definition}

\begin{definition}
\label{summablecompactmeasurable}
Let $T$ be an operator affiliated with $\Nn$, $0<p<\infty$, and $x>0$. Then
$T$ is said to be 
\begin{itemize}
 \item 
\textbf{$p$-summable} if 
\[
 \left\lVert T\right\rVert_p := \tau(|T|^p)^{1/p} < \infty \mbox{ ,}
\]
\item
\textbf{$\tau$-compact} 
if \[\lim_{x\rightarrow\infty} \mu_x(T)=0\mbox{ ,}\]
\item
\textbf{$\tau$-measurable} if
for each $\varepsilon>0$ there exists a projection $E\in\Nn$ such that
\[ \operatorname{Ran}(E) \subset
\operatorname{Dom}(T)
\hspace{0.3cm} \mbox{  and  } \hspace{0.3cm} \tau(1-E) < \varepsilon \mbox{ .} \]
\end{itemize}
\end{definition}

\begin{remark}
 Anything in $\Nn$ is $\tau$-measurable. If a self-adjoint operator $T$ is affiliated with $\Nn$ and its resolvent is  $\tau$-compact,
 then $T$ is $\tau$-measurable \cite{type2index}.
\end{remark}

\begin{proposition}[\cite{snumber}]
\label{muproperties}
 Let $T$, $S$, $R$ be $\tau$-measurable operators.

\begin{enumerate}
\item 
The map: $x\in (0,\infty) \rightarrow \mu_x(T)$ is non-increasing and continuous from the right. Moreover,
\[
 \lim _{x\searrow 0 } \mu_x(T)=\left\lVert T\right\rVert\in[0,\infty] \mbox{ .}
\]
\item 
$\mu_x(T)=\mu_x(|T|)=\mu_x(T^*)$ and $\mu_x(zT)=|z|\mu_x(T)$ for $x>0$ and $z\in\mathbb{C}$.
\item 
$\mu_x(T)\leq \mu_x(S)$, $x>0$, if $0\leq T\leq S$.
\item 
$\mu_x(f(|T|))=f(\mu_x(|T|))$, $x>0$ for any continuous increasing function $f$ on $[0,\infty)$ with $f(0)\geq0$.
\item 
$\mu_x(STR) \leq \left\lVert S\right\rVert \left\lVert R\right\rVert  \mu_x(T) $, $x>0$.
\end{enumerate}

\end{proposition}
\begin{proposition}[\cite{snumber}]
\label{taupositive}
 Let $T$ be a positive $\tau$-measurable operator. Then
\[
 \tau(T)=\int_0^\infty \mu_x(T)dx \mbox{ .}
\]

\end{proposition}

\begin{proposition}[\cite{snumber}]
\label{pnormproduct}
 Let $T$, $S$, and $R$ be operators in $\Nn$. Then for  $0<p<\infty$,
\[
\left\lVert STR \right\rVert_p \leq \left\lVert S \right\rVert \left\lVert R\right\rVert \left\lVert T \right\rVert_p
\mbox{ .}
\]
\end{proposition}

Denote by $\Ll^p$ the space of all
 $p$-summable operators in $\Nn$. For $0<p<\infty$, the space
$\Ll^p$ forms a norm closed two-sided ideal in $\Nn$ with norm given by 
$\left\lVert \cdot  \right\rVert_p + \left\lVert \cdot \right\rVert $.
Denote by $\Kk$ the space of all $\tau$-compact operators in $\Nn$.
 The space $\Kk$ forms a norm closed two-sided ideal in $\Nn$.

\begin{theorem}[\cite{snumber}]
\label{holder}
 Let $T$, $S$ be $\tau$-measurable operators. Then 
\begin{enumerate}
 \item 
 $\left\lVert TS\right\rVert _r \leq \left\lVert T \right\rVert _p \left\lVert S \right\rVert _q$
 for  $p,q,r >0$ and $p^{-1}+q^{-1}=r^{-1}$.
\item 
$\left\lVert T+S\right\rVert _p \leq \left\lVert T\right\rVert _p +\left\lVert S \right\rVert _p$
 for $p\geq 1$.
\end{enumerate}
\end{theorem}

For $\Nn=B(\Hh)$ with $\tau$ the operator trace,
then $p$-summability and $\tau$-compactness are the usual 
notion of $p$-summability and compactness, and the
ideals $\Ll^p$ and $\Kk$ are the usual ideal of Schatten $p$-class
and the ideal compact operators.


\bibliographystyle{plain}

\bibliography{thesis}

\end{document}